\documentclass[letterpaper,10pt,journal]{IEEEtran}

\usepackage{amsthm}
\usepackage{amsmath}
\usepackage{color}
\usepackage{mathtools}
\usepackage{amssymb}
\usepackage{amsfonts}
\usepackage{amstext}
\usepackage{graphicx}
\usepackage{enumerate}
\usepackage{algorithm,algpseudocode}
\usepackage{balance}
\usepackage{cite}
\usepackage{bm}
\usepackage{tikz}
\usepackage{accents}
\usetikzlibrary{shapes,arrows}
\usepackage{verbatim}
\usepackage{psfrag}
\usepackage[colorlinks=true, urlcolor=blue, citecolor=blue, linkcolor=blue]{hyperref}
\usepackage{multirow}
\usepackage{subfig}
\newcommand{\source}{{THIS IS A PREPRINT VERSION. IF YOU FOUND THIS READING USEFUL FOR YOUR RESEARCH PLEASE CITE THE PUBLISHED VERSION DOI: \href{https://doi.org/10.1109/TITS.2020.2975815}{https://doi.org/10.1109/TITS.2020.2975815}}}

\newcounter{Theorem}
\newtheorem{theorem}{Theorem}[Theorem]

\newcounter{Lemma}
\newtheorem{lemma}{Lemma}[Lemma]

\newcounter{Definition}
\newtheorem{definition}{Definition}[Definition]

\newcounter{Assumption}
\newtheorem{assumption}{Assumption}[Assumption]

\newcounter{Problem}
\newtheorem{problem}{Problem}[Problem]

\newcounter{Proposition}
\newtheorem{proposition}{Proposition}[Proposition]

\newcounter{Remark}
\newtheorem{remark}{Remark}[Remark]
\DeclareMathOperator{\dom}{dom}
\DeclareMathOperator{\R}{\mathbb{R}}
\DeclareMathOperator{\N}{\mathbb{N}}
\DeclareMathOperator{\L2t}{\mathcal{L}_{2t}}
\DeclareMathOperator{\tL2}{\textit{t}\text{-}\mathcal{L}_2}

\usepackage{fancyhdr}
\begin{document}

\bstctlcite{IEEEexample:BSTcontrol}

\renewcommand{\baselinestretch}{0.98}
\tikzstyle{block} = [draw, rectangle, 
    minimum height=2.5em, minimum width=2.5em]
\tikzstyle{sum} = [draw, circle, minimum width=5pt, draw, inner sep=0pt, node distance=1cm]
\tikzstyle{joint} = [draw, circle, minimum width=2pt, fill, inner sep=0pt, node distance=1cm]
\tikzstyle{point} = [draw, circle, minimum width=0pt, inner sep=0pt, node distance=1cm]
\tikzstyle{input} = [coordinate]
\tikzstyle{output} = [coordinate]
\makeatletter
\def\ps@IEEEtitlepagestyle{}
\title{A Hybrid Controller for DOS-Resilient String-Stable Vehicle Platoons}

\author{Roberto~Merco$^{\dagger}$,~\IEEEmembership{Student Member,~IEEE,}
        Francesco~Ferrante$^{\ddagger}$,~\IEEEmembership{Member,~IEEE,}
        and~Pierluigi~Pisu$^{\dagger}$,~\IEEEmembership{Member,~IEEE}% <-this % stops a space
\thanks{$^{\dagger}$Automotive Department of Clemson University, Greenville SC 29607 USA. Email: rmerco@g.clemson.edu, pisup@clemson.edu.}%
\thanks{$^{\ddagger}$Univ. Grenoble Alpes, CNRS, GIPSA-lab, F- 38000 Grenoble, France. Email: francesco.ferrante@gipsa-lab.fr}
\thanks{This material is based upon work supported by the National Science
Foundation (NSF) under grant No. CNS-1544910. Research by Francesco Ferrante is funded in part by ANR via project HANDY, number ANR-18-CE40-0010. Any opinions, findings and conclusions or recommendations expressed in this material are those of the authors and do not necessarily reflect the views of the National Science Foundation.}
}

\maketitle

\begin{abstract}
This paper deals with the design of resilient Cooperative Adaptive Cruise Control (CACC) for homogeneous vehicle platoons in which communication is vulnerable to Denial-of-Service (DOS) attacks. We consider DOS attacks as consecutive packet dropouts. We present a controller tuning procedure based on linear matrix inequalities (LMI) that maximizes the resiliency to DOS attacks, while guaranteeing performance and string stability. The design procedure returns controller gains and gives a lower bound on the maximum allowable number of successive packet dropouts. A numerical example is employed to illustrate the effectiveness of the proposed approach. 
\end{abstract}


\IEEEpeerreviewmaketitle

\section{Introduction}
\subsection{Background}
While the demand for mobility is growing over the years, Intelligent Transportation Systems (ITS) is a promising advanced solution capable of improving the efficiency and safety of the whole mobility infrastructure via Connected and Autonomous Vehicles (CAVs) \cite{lu2014connected}.

Probably the most famous application involving connectivity between vehicles is the Cooperative Adaptive Cruise Control (CACC). 
Common metrics for CACC are \textit{individual vehicle stability} and \textit{string stability}. The former refers to the reduced distance between vehicles and leads to higher traffic throughput and higher fuel economy \cite{van2006impact}. The latter refers to the attenuation of disturbances and shock waves throughout the string of cars and enables to improve traffic flow by avoiding the so-called phantom traffic jam \cite{van2006impact,ploeg2014lp}. 
By employing inter-vehicle communication (IVC) and on-board sensors, e.g., radar and lidar, CACC improves features of the Adaptive Cruise Control (ACC) by reducing the distance between vehicles and by providing enhanced string stability properties. Indeed, CACC can guarantee string stability with inter-vehicle time gaps smaller than one second, whereas ACC cannot \cite{ploeg2014lp,naus2010string}. Notice that a small inter-vehicle time gap leads to higher traffic throughput and higher fuel economy \cite{van2006impact}.

Besides the benefits introduced by IVC, this wireless network exposes vehicle control systems to network induced imperfections, i.e., packet dropping, network delays \cite{heemels2010networked}, and network vulnerabilities, as e.g., cyber attacks \cite{dadras2015vehicular}. Such an unreliable and compromised network can affect vehicle platoons by leading to loss of performance and safety-related issues, e.g., collisions, which could lead to loss of human lives \cite{oncu2014cooperative,acciani2018cooperative,dadras2015vehicular}. 

This paper deals with the design of CACC resilient to DOS attacks. DOS-resilient design approaches seek to maximize the number of consecutive packet dropouts that can be tolerated to maintain stability and performance of the vehicle platooning. Indeed, differently from natural packet dropouts that generally generate random pattern of packet dropouts \cite{mir2014lte,rayamajhi2018impact}, DOS attacks maliciously generate a prolonged period without communication with the purpose of disrupting the control system \cite{amin2009safe,dolk2017eventDoS,yuan2013resilient}. As such, the use of probabilistic approaches is ineffective in this scenario.

\subsection{Related Work} \label{sec:literature}
Several results can be found in the literature on the modeling and design of vehicle platoons with CACC controllers, as it emerges from \cite{naus2010string,oncu2014cooperative,li2015overview,ploeg2011design,ploeg2014controller} and references therein. 
Ploeg et al. \cite{ploeg2011design} introduce one of the most relevant designs for CACC. However, such a design is based on continuous-time control techniques that do not consider the discrete packet-based nature of the IVC. 
To this end, the research community analyzed and enhanced the CACC in \cite{ploeg2011design} by considering an unreliable packet-based network. In \cite{oncu2012string}, hybrid modeling is used to analyze stability of vehicle platoons controlled by CACC, where IVC is employed by adopting scheduling protocols. 
In \cite{gong2018sampled} a sampled-data design approach is proposed, whereas \cite{li2019string,dolk2017event} discuss event-triggered control strategies.
To cope with communication losses, \cite{ploeg2014graceful,wu2019cooperative} propose CACC algorithms that replace data exchanged between vehicles by employing estimation techniques based only on onboard sensors. To address the same issue, Harfouch et al. \cite{harfouch2017adaptive}, instead, propose a switching strategy between CACC and ACC. Notice that these approaches consist of some form of fallback to control strategies that avoid exchanging information between vehicles. Therefore, they require inter-vehicle time gaps larger than those achievable with traditional CACC to guarantee string stability. This affects fuel consumption and traffic throughput.

While the design of stable vehicle platoons with cars connected via an unreliable network has been widely investigated, the design of control systems for connected vehicles in the presence of cyber attacks is still an active research area despite being significantly critical. Indeed, compromised networks lead to severe consequences with the involvement of human lives and safety; see, i.e., Dadras et al. \cite{dadras2015vehicular}. DOS jamming attack in connected vehicles is investigated in \cite{alipour2017string}, where string stability is analyzed under packet dropping generated by jamming actions. Biron et al. \cite{biron2018real} propose and approach to detect and estimate the entity of DOS attacks in a vehicle platoon by modeling it as an unknown constant delay and propose a control architecture able to mitigate the effect of the attack. Dolk et al. \cite{dolk2017eventDoS} employ the CACC as a case study for evaluating a design procedure for DOS-resilient event-triggered mechanisms.

\subsection{Contributions}

The research community has been improving CACC in \cite{ploeg2011design} by introducing event-driven communication or fallback strategies to obtain safer platoons in case of an unreliable network. However, to the best of the authors knowledge, none of the existing works propose a tuning approach for the CACC in \cite{ploeg2011design} to increase the resilience to consecutive communication losses. Such a tuning could increase the resilience of the CACC to DOS attacks and unreliable networks. Indeed, it could be used with event-driven or fallback strategies to enhance the resilience of vehicle platoons even further. In particular, it would allow having a longer attack tolerant time interval that leads to postponing or improving fallback to safer strategies.
Furthermore, it emerges from the literature review that few control strategies deal with vehicle performance in regulating the distance gap between cars while improving resiliency.

In this paper, we focus on addressing these research gaps by proposing a design approach for CACC that aims at maximizing the resiliency to DOS attacks while guaranteeing required performance and string stability. 
In particular, the main contributions of this paper are listed as follows: 
\begin{itemize}
\item We propose a decentralized hybrid controller that modifies the continuous-time proportional-derivative (PD) regulator in  \cite{ploeg2011design} by adding a Zero Order Hold (ZOH) device that allows taking into account the packet-based nature of the IVC. 
\item We devise a numerically efficient offline tuning algorithm based on linear matrix inequalities (LMI). Such an algorithm allows finding optimal gains for the controller capable of maximizing the resiliency to DOS while guaranteeing performance requirements and string stability of the vehicle platooning.
\item Given the controller gains, the proposed algorithm estimates the value of the \textit{maximum allowable number of successive packet dropouts} (MANSD) that identifies the worst DOS attack that the control system can overcome without compromising string stability.
\end{itemize}

This paper extends our prior works in \cite{MercoACC,MercoIV19}. Here we extensively present the problem and the methodology by providing details, proofs, and simulations that are not present in \cite{MercoACC,MercoIV19} due to space limitation.

The remainder of this paper is organized as follows. Section II introduces the model of vehicle platoons. Section III formulates the control problem. Next, Section IV presents the stability analysis and the proposed design algorithm. Simulation results are shown in Section V, and Section VI concludes the paper.

\subsection{\textbf{Notation}}
The set $\mathbb{N}$ is the set of the strictly positive integers, $\mathbb{N}_0=\mathbb{N}\cup\{0\}$, $\R$ is the set of real numbers, $\mathbb{R}_{\geq 0}$, $\R_{>0}$, $\R_{<0}$ are, respectively, the sets of nonnegative, positive, and negative real numbers, $\mathbb{C}$ is the set of complex numbers. Given $z\in\mathbb{C}$, $\Re(z)$ and $\Im(z)$ denote, respectively, the real and the imaginary part of $z$. For any given real polynomial $\rho$, $\Lambda_{\max}(\rho)=\displaystyle\max_{s\colon \rho(s)=0}\Re(s)$ and $\Im(\rho)=\displaystyle\{\omega\in\mathbb{R}\colon\exists \alpha\in\mathbb{R}\,\,\text{s.t.}\,\, \rho(\alpha+j\omega)=0\}$ (the set of the imaginary parts of the roots of $\rho$). With a slight abuse of notation, given a real polynomial $\rho(s)=s^2+as+b$, with $a, b\in\mathbb{R}_{>0}$, we denote the damping ratio of the roots of $\rho$ as $\zeta(\rho)=\frac{a}{2\sqrt{b}}$. The symbol $\mathbb{R}^n$ represents the Euclidean space of dimension $n$, $\mathbb{R}^{n\times m}$ is the set of $n \times m$ real matrices. Given $A\in \mathbb{R}^{n \times m}$, $A^\top$ denotes the transpose of $A$, and when $n=m$, $A^{-\top}=(A^\top)^{-1}$ (when $A$ is nonsingular), $\text{He}(A)=A+A^\top$, $\Lambda(A)$ denotes the spectrum of $A$, $\Lambda_{\max}(A)\coloneqq\max\Re(\Lambda(A))$, and $\zeta_{\min}(A)\coloneqq\min\zeta(\Lambda(A))$.   
For a symmetric matrix $A$, $A>0$, $A \geq 0$, $A<0$, and $A \leq 0$ means that $A$ is, respectively, positive definite, positive semidefinite, negative definite, and negative semidefinite. The symbol $S_+^n$ represents the set of $n \times n$ symmetric positive definite matrices, $\lambda_{\min}(A)$ and $\lambda_{\max}(A)$ denote respectively the smallest and the largest eigenvalue of the symmetric matrix $A$. In partitioned symmetric matrices, the symbol $\bullet$ represent a symmetric block. For a vector $x \in \mathbb{R}^n$, $\vert x \vert$ denotes the Euclidean norm. Given two vectors $x$ and $y$, we denote $(x,y)=[x^\top,y^\top]^\top$. Given a vector $x \in \mathbb{R}^n$ and a closed set $\mathcal{A}$, the distance of $x$ to $\mathcal{A}$ is defined as $\vert x \vert_\mathcal{A}=\textrm{inf}_{y\in\mathcal{A}} \vert x-y \vert$. For any function $z: \mathbb{R} \rightarrow \mathbb{R}^n$, we denote $z(t^+):=\textrm{lim}_{s\rightarrow t^+} z(s)$, when it exists. A function $\alpha\colon\mathbb{R}_{\geq 0} \rightarrow \mathbb{R}_{\geq 0}$ is said to be of class $\mathcal{K}$ if it is continuous, strictly increasing, and $\alpha(0)=0$.

\section{Modeling}
\subsection{Platooning Dynamics}
We consider a homogeneous vehicle platooning formed by $m$ identical cars. The vehicle that leads the platoon is denoted by $\mathcal{V}_0$, whereas $\mathcal{V}_i$, $i\in P_m\coloneqq\{1,2, \dots, m\}$, identifies the following vehicles. As a target of the vehicle platooning, $\mathcal{V}_i$ must maintain the reference distance $d_{r_i}$ from its preceding vehicle $\mathcal{V}_{i-1}$ by employing a constant time gap policy reference. In particular, $d_{r_i}= r_i + h \, v_i, \, i \in P_m$, where $r_i$ and $v_i$ are, respectively, the standstill reference distance and the speed of $\mathcal{V}_i$, and $h\in\mathbb{R}_{> 0}$ is the constant time gap between vehicles. The spacing error $e_i$ is given by
\begin{equation} 
\label{eq:spacingError}
e_i\coloneqq(q_{i-1}-q_i - L) - (r_i + h \, v_i) = d_{i} - d_{r_i} 
\end{equation}
where $q_i$, $L$, and $d_i\coloneqq q_{i-1}-q_i - L$ denote, respectively, the position, the length of vehicles in the platoon, and the distance between vehicles $\mathcal{V}_i$ and $\mathcal{V}_{i-1}$.

Each vehicle in the platooning can be modeled as a continuous-time linear time-invariant dynamical system; see \cite{stankovic2000decentralized} for further details. 
In particular, we consider:
\begin{equation} \label{eq:leader}
\mathcal{V}_0:
\begin{pmatrix} 
\dot{v}_0 \\
\dot{a}_0 \\
\end{pmatrix}
= 
\begin{pmatrix} 
a_0 \\
-\frac{1}{\tau_d} a_0+\frac{1}{\tau_d} u_0 \\
\end{pmatrix}
\end{equation}
\begin{equation} \label{eq:vehicleDynamics}
\mathcal{V}_i:
\begin{pmatrix} 
\dot{e}_i \\
\dot{v}_i \\
\dot{a}_i \\
\end{pmatrix}
= 
\begin{pmatrix} 
v_{i-1}-v_i-h a_i \\
a_i \\
-\frac{1}{\tau_d} a_i+\frac{1}{\tau_d} u_i \\
\end{pmatrix}, \quad i \in P_m
\end{equation}
where $v_i$ ($v_0$), $a_i$ ($a_0$), and $u_i$ ($u_0$) are, respectively, the speed, the acceleration, and the control input of $\mathcal{V}_i$ ($\mathcal{V}_0$), and $\tau_d$ represents the time constant of the powertrain dynamics of the vehicles in the platooning. Since we consider a homogeneous vehicle platooning, $h$ and $\tau_d$ are identically for each vehicle.

In this paper, we modify the decentralized CACC controller in \cite{ploeg2011design} by considering the impulsive behavior of the network induced by intermittent communication. Such a controller is given by the following dynamics
\begin{equation} \label{eq:TCcontroller}
\dot{u}_i = \mathcal{K}_{h}(u_i,\chi_i), \, \chi_i=\mathcal{K}_{PD}(e_i) + u_{i-1}
\end{equation}
where $u_{i-1}$ is the control signal of $\mathcal{V}_{i-1}$, $\mathcal{K}_{h}(u_i,\chi_i) \coloneqq -\frac{1}{h} u_i + \frac{1}{h} \chi_i$, and $\mathcal{K}_{PD}(e_i)\coloneqq k_p e_i + k_d \dot{e}_i$, and $k_p$ and $k_d$ are the controller gains. Observe that each vehicle is controlled by a dynamic controller as in \eqref{eq:TCcontroller}, where gains $k_p$ and $k_d$ are independent from the vehicle's index. Controller \eqref{eq:TCcontroller} is designed to rely on on-board sensors (e.g., radars, lidars and accelerometers) for measurements of $e_i$ and $\dot{e}_i$, and IVC for the signal $u_{i-1}$ of the preceding vehicle $\mathcal{V}_{i-1}$. However, in \cite{ploeg2011design}, the remote signal $u_{i-1}$ is treated as a continuous-time signal even though it is shared through a network. In this paper, we modify the structure of the controller \eqref{eq:TCcontroller} by proposing a hybrid controller able to deal with the discrete behavior of the packet-based network communication used to share the value of $u_{i-1}$. 
We realistically assume that measurements from on-board sensors are available continuously over time. A similar assumption is used in \cite{oncu2012string}.

\subsection{Communication Network and DOS Attacks}
We assume that the measurement $u_{i-1}$ is sampled and sent periodically at instants $t_{k_{i-1}}, \, k_{i-1}\in \N_0$ with constant transmission interval $T_s\in\R_{>0}$, i.e., $t_{k_{i-1}+1}-t_{k_{i-1}}=T_s$, $t_0=0$. In addition, we consider that the presence of IVC exposes the vehicle platooning to DOS attacks. DOS attacks aim at making the network unavailable, e.g., by injecting excessive traffic \cite{yuan2013resilient,amin2009safe} or by adopting jamming strategies \cite{poisel2011modern}. Hence, under DOS attacks, the IVC experiences packet dropouts and connected vehicles are unable to cooperate properly \cite{amoozadeh2015security}. Indeed, during DOS attacks, vehicles in the platooning do not receive any signal from the preceding vehicles, hence, $u_{i-1}$ is unavailable to the controllers.

\subsection{Adversarial Model of DOS attacks} \label{sec:adversarial}

In this paper, we consider that the objective of the attacker is to perform a DOS attack by using a  radio jamming strategy, which deliberately disrupts communications over a geographic area \cite{amoozadeh2015security}. The jamming strategy is unknown to the controller, but it is assumed that it is energy and geography constrained. Indeed, the attacker can perform DOS attacks that generate packet dropouts only for a finite period in time due to the limited amount of resources and because the platoon could move to an attack-free area. Notice that detection and mitigation techniques could be implemented in the vehicle network to reduce the duration of the DOS attacks \cite{amoozadeh2015security,hamieh2009detection}.
As such, packet losses due to DOS attacks can be assumed to be persistent only for a limited period of time \cite{amin2009safe}, and can be modeled by considering an upper bound to the maximum number of successive packet dropouts. 

Inspired by \cite{dolk2017eventDoS,feng2017resilient}, we model a DOS attack as a limited time period where an attacker succeeds in blocking the signal $u_{i-1}$ in such a way that it cannot reach the controller in vehicle $\mathcal{V}_{i}$. Therefore, several DOS attacks can be seen as a sequence of intervals $\{H_n\}_{n\in\mathbb{N}}$, each one of finite length, where the IVC network is interrupted. Specifically, we assume that the $n$-th DOS attack produces $\iota_n \in \{0,1, \ldots,\Delta\}, \forall n\in\mathbb{N}$, consecutive packet dropouts, where $\Delta \in \N_0$ is the \textit{maximum allowable number of successive packet dropouts} (MANSD), i.e., the maximum number of consecutive packet losses such that the vehicle platooning maintains his stability properties. Furthermore, we assume that at least one successful transmission is expected to occur in between intervals $\{H_n\}_{n\in\mathbb{N}}$. 

Notice that prolonged unavailability of network data can degrade or compromise string stability; see  \cite{oncu2012string,kester2014critical,oncu2014cooperative}. Therefore, tuning the CACC in such a way to maximize the resilience to DOS attacks is of paramount importance. In this paper, this is achieved by selecting the controller gains to maximize the value of MANSD. Moreover, notice that the estimation of MANSD is significant as the resiliency metric.

\section{Controller Outline and Problem Formulation}
\subsection{Proposed Networked Controller}
In this paper, we propose a modified version of the controller in \eqref{eq:TCcontroller} that takes into account the discrete nature of the data available through the network. The control scheme is depicted in Fig. \ref{fig:controlScheme}.

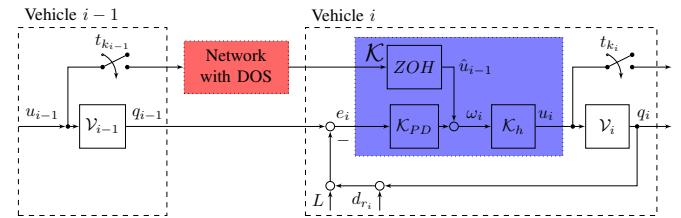
\begin{figure}[hpbt]
\resizebox{\hsize}{!}{
\begin{tikzpicture}[auto, node distance=2cm,>=latex']
 
 % Nodes central row
    \node [input, name=input] {};
    \node [joint,name=u_im1,right of =input,node distance=1cm] {};
    \node [block, right of=u_im1,node distance=0.7cm] (v_im1) {$\mathcal{V}_{i-1}$};
    \node [sum, right of=v_im1,node distance=4.6cm] (sume_i) {};
    \node [joint,name=ei_joint,right of =sume_i,node distance=0.7cm] {};
    \draw [dotted, fill=blue!50] (sume_i) ++(0.5cm,1.8cm) rectangle ++(4.2cm,-2.4cm) ; %
    \node [block, right of=ei_joint,node distance=1cm] (K_pd_i) {$\mathcal{K}_{PD}$};
    \node [sum, right of=K_pd_i, node distance=0.8cm] (sumx_i) {};
    \node [block, right of=sumx_i,node distance=1.2cm] (H_m1) {$\mathcal{K}_{h}$};
    \node [joint,name=u_i,right of =H_m1,node distance=1.2cm] {};
    \node [block, right of=u_i,node distance=0.7cm] (v_i) {$\mathcal{V}_i$};
    \node [joint,name=q_i,right of =v_i,node distance=0.6cm] {};
    \node [output, right of=v_i,node distance=1.3cm] (outputq_i) {};

	% Nodes bottom row
    \node [sum, below of=sume_i,node distance=1.2cm] (sumr) {};
    \node [input, below of=sumr,node distance=0.5cm](inputr) {};
    \node [sum, right of=sumr] (sumL) {};
    \node [input, below of=sumL,node distance=0.5cm](inputL) {};
	
	% Nodes upper row
	\node [block, above of=K_pd_i,node distance=1.2cm] (zoh) {$ZOH$};
	\node [left of=zoh,node distance=0.8cm, above] {{\Large $\mathcal{K}$}};
	\node [joint,name=virtualSwitch1,above of =v_i,node distance=1.2cm] {};
	\node [joint,name=virtualSwitch2,right of =virtualSwitch1, node distance=0.5cm] {};
	\node [joint,name=virtualSwitch3,above of =virtualSwitch2, node distance=0.25cm] {};
	\node [output, above of=outputq_i,node distance=1.2cm] (outputu_i) {};
	\node [joint,name=virtualSwitch1m1,above of =v_im1,node distance=1.2cm] {};
	\node [joint,name=virtualSwitch2m1,right of =virtualSwitch1m1, node distance=0.5cm] {};
	\node [joint,name=virtualSwitch3m1,above of =virtualSwitch2m1, node distance=0.25cm] {};
  	\node [block, dotted, fill=red!60, right of=virtualSwitch2m1,node distance=2.2cm] (net) {\begin{tabular}{c} Network \\ with DOS \end{tabular}};

	%for arc switch	
	\node [point,name=Arc1,above of =virtualSwitch1,node distance=0.25cm,label={[label distance=-0.3cm]10:$t_{k_i}$}] {};
	\node [point,name=Arc1m1,above of =virtualSwitch1m1,node distance=0.25cm,label={[label distance=-0.3cm]10:$t_{k_{i-1}}$}] {};
    
    % link central row
    \draw [->] (input) -- node {$u_{i-1}$} (u_im1);
    \draw [->] (u_im1) -- node {} (v_im1);
    \draw [->] (v_im1) -- node[pos=0.11] {$q_{i-1}$} (sume_i);
    \draw [->] (sume_i) -- node[pos=0.15] {$e_i$} (K_pd_i);
    \draw [->] (K_pd_i) -- node {} (sumx_i);
    \draw [->] (sumx_i) -- node {$\omega_i$} (H_m1);
    \draw [-] (H_m1) -- node[pos=0.3]{$u_i$} (u_i);
	\draw [->] (u_i) -- node {} (v_i);
	\draw [-] (v_i) -- node {} (q_i);
	\draw [->] (q_i) -- node[pos=0.15] {$q_i$}(outputq_i);
    
	% link bottom row
    \draw [->] (q_i) |- (sumL);
    \draw [draw,->] (inputL) -- node {$d_{r_i}$} (sumL);
    \draw [->] (sumL) -- node[pos=0.95] {} 
        node [near end] {} (sumr);
    \draw [draw,->] (inputr) -- node {$L$} (sumr);
    \draw [->] (sumr) -| node[pos=0.9,right] {$-$} 
        node [near end] {} (sume_i);
        
    % link upper row
    \draw [draw,-] (u_i) |- node {} (virtualSwitch1);
    \draw [draw,-] (virtualSwitch1) -- node {} (virtualSwitch3);
    \draw [->] (virtualSwitch2) -- (outputu_i);
    \draw [->] (Arc1.east) to [out=-30,in=90] +(0.25cm,-0.5cm);
    \draw [draw,-] (u_im1) |- node {} (virtualSwitch1m1);
    \draw [draw,-] (virtualSwitch1m1) -- node {} (virtualSwitch3m1);
  	\draw [->] (Arc1m1.east) to [out=-30,in=90] +(0.25cm,-0.5cm);
    \draw [->] (virtualSwitch2m1) -- (net);
    \draw [->] (net) -- node[pos=0.35,above]{}(zoh.west);
    \draw [->] (zoh) -| node[pos=0.45,right]{$\hat{u}_{i-1}$}(sumx_i);

    % dashed nodes
    \draw [dashed] (sume_i) ++(-0.5cm,2cm) rectangle ++(7.1cm,-3.8cm) node[pos=0] {Vehicle $i$};
    \draw [dashed] (input) ++(0cm,2cm) rectangle ++(3cm,-3.8cm) node[pos=0] {Vehicle $i-1$};

\end{tikzpicture}
}
\caption{Schematic of the control system in $\mathcal{V}_i$. In particular, signal connections between vehicles $\mathcal{V}_{i-1}$ and $\mathcal{V}_i$ along with the proposed hybrid controller $\mathcal{K}$ (in the purple area) are shown.}
\label{fig:controlScheme}
\end{figure}

For each $i$, the proposed hybrid controller handles discrete measurements of $u_{i-1}$, which are available through network packets only at time $t_{k_{i-1}}, \, k_{i-1} \in \N_0$. Specifically, the controller in \eqref{eq:TCcontroller} is augmented with a memory state $\hat{u}_{i-1}\in \R$, which stores the last received value of $u_{i-1}$. 
When at time $t_{k_{i-1}}$ a new measurement is available through the network, $\hat{u}_{i-1}$ is instantaneously reset to $u_{i-1}(t_{k_{i-1}})$. In between received measurements, $\hat{u}_{i-1}$ is kept constant in a ZOH fashion.
More precisely, dynamics of $\hat{u}_{i-1}$ can be modeled as a system with jumps in its state. In particular, its dynamics are as follows for all $k_{i-1} \in \N_0$:
\begin{equation} \label{eq:ZOH}
\scalebox{0.9}{$
\left\{\begin{array}{ll}
\dot{\hat{u}}_{i-1}(t) = 0 & \forall \, t \not= t_{k_{i-1}} \, \text{or} \,\, t = t_{k_{i-1}} \in \bigcup_{n\in\mathbb{N}} H_n \\
\hat{u}_{i-1}(t^+) = u_{i-1}(t_{k_{i-1}}) &  \forall \, t = t_{k_{i-1}} \notin \bigcup_{n\in\mathbb{N}} H_n
\end{array} \right.$}
\end{equation}
Notice that $\hat{u}_{i-1}$ is set to $u_{i-1}(t_{k_{i-1}})$ only in case of successful transmissions. 

Differently from  \eqref{eq:TCcontroller}, the controller is fed with $\hat{u}_{i-1}$, and its continuous-time dynamics are given by:
\begin{equation} \label{eq:controller}
\dot{u}_i = \mathcal{K}_{h}(u_i,\omega_i), \quad \omega_i\coloneqq \mathcal{K}_{PD}(e_i) + \hat{u}_{i-1}
\end{equation}
The interconnection between the ZOH device in \eqref{eq:ZOH} and the controller in \eqref{eq:controller} is denoted by $\mathcal{K}$ and represents the proposed hybrid controller.

\subsection{Hybrid Modeling}
The stability of the vehicle platooning is studied by analyzing the dynamics of the closed-loop system obtained by the interconnection of \eqref{eq:vehicleDynamics}, \eqref{eq:ZOH}, and \eqref{eq:controller}.  To this end, let $e_{1,i} := e_i$, $e_{2,i} := \dot{e}_i$ and $e_{3,i} := \ddot{e}_i$, by straightforward calculations, one has that for all $k_{i-1}\in\N_0$
\begin{equation} \label{eq:impulsiveSys}
\scalebox{0.9}{
$
\begin{array}{ll}
\!\!\!\!\! \left\{  \!
\begin{aligned}
& \dot{e}_{1,i} = e_{2,i} \\
& \dot{e}_{2,i} = e_{3,i}\\
& \dot{e}_{3,i} \! = \! - \! \frac{k_p}{\tau_d}  e_{1,i} \!\! - \!\! \frac{k_d}{\tau_d}  e_{2,i} \!\!-\!\!\frac{1}{\tau_d} e_{3,i} \!\! +\!\! \frac{1}{\tau_d} u_{i-1} \!\! - \!\!\frac{1}{\tau_d} \hat{u}_{i-1}\\
& \dot{u}_{i-1} = -\frac{1}{h} u_{i-1} \! + \! \frac{1}{h} \omega_{i-1} \\
& \dot{\hat{u}}_{i-1} = 0
\end{aligned}\right. & 
\begin{aligned}
&\forall \, t \! \not=\! t_{k_{i-1}} \,\, \text{or} \\
& t \! = \! t_{k_{i-1}} \!\!\in\!\! \bigcup_{n\in\mathbb{N}}\!\! H_n \\
\end{aligned}\\[1.7cm]
\!\!\!\!\! \left\{  \!
\begin{aligned}
& e_{1,i}(t^+) = e_{1,i}(t) \\
& e_{2,i}(t^+) = e_{2,i}(t)\\
& e_{3,i}(t^+) = e_{3,i}(t)\\
& u_{i-1}(t^+) = u_{i-1}(t) \\
& \hat{u}_{i-1}(t^+) = u_{i-1}(t)
\end{aligned}\right. & 
\begin{aligned}
& \forall \, t\! =\! t_{k_{i-1}}\!\! \notin \!\!\bigcup_{n\in\mathbb{N}}\!\! H_n\\
\end{aligned}\\
\end{array}    
$}
\end{equation}
For the sake of notation, notice that the dependence on time in continuous-time dynamics is omitted.
Due to the hybrid controller and network behaviors, such error dynamics are characterized by the interplay of differential equations and instantaneous jumps. Furthermore, given the aperiodicity and unpredictability of the successful transmissions, the stability of such an impulsive model is difficult to analyze via traditional tools.
Therefore, we model system \eqref{eq:impulsiveSys} into the hybrid systems framework in \cite{goebel2012hybrid}, which provides useful tools to address  stability analysis for hybrid systems.
To this end, we introduce the auxiliary variable $\sigma_{i-1} \in \mathbb{R}_{\geq 0}$, which models the hidden time-driven mechanism that triggers jumps in the controller when a new packet is received from the network. In particular, from \eqref{eq:impulsiveSys}, one obtains the following hybrid system for the error dynamics
\begin{equation}\label{eq:firstModel}
\begin{aligned}
\left\lbrace
\begin{array}{ll}
\dot{\xi}_i = f_{\xi}(\xi_i,\omega_{i-1}) & \xi_i \in C_{\xi}, \omega_{i-1} \in \mathbb{R} \\
\xi_i^+ = g_{\xi}(\xi_i) & \xi_i \in D_{\xi}
\end{array}
\right. 
\end{aligned}
\end{equation} 
where $\xi_i:=(e_{1,i},e_{2,i},e_{3,i},u_{i-1},\hat{u}_{i-1},\sigma_{i-1}) \in \mathbb{R}^{6}$ defines the state of the hybrid system, and $\omega_{i-1}$ is the input.

By following the formalism introduced in \cite{goebel2012hybrid}, $\dot{\xi}_i$ stands for the velocity of the state and $\xi_i^+$ indicates the value of the state after an instantaneous change due to received network packets. The set where the continuous evolution (flow) of the state occurs is named flow set and it is defined as 
$
C_{\xi} \coloneqq  \R^5 \times [0,(\Delta+1)T_s]
$.
According to the definition of the flow set, the state $\xi_i$ can evolve by following the flow dynamics whenever the variable $\sigma_{i-1}\in[0,(\Delta+1)T_s]$.
The flow dynamics follow the differential equation $\dot{\xi}_i = f_{\xi}(\xi_i,\omega_{i-1})$, where $f_{\xi}$ is named flow map and it is defined as 
$
f_{\xi}(\xi_i,\omega_{i-1}) := ( e_{2,i}, e_{3,i}, - \frac{k_p}{\tau_d}  e_{1,i} - \frac{k_d}{\tau_d}  e_{2,i} -\frac{1}{\tau_d} e_{3,i} \! +\! \frac{1}{\tau_d} u_{i-1}  -\frac{1}{\tau_d} \hat{u}_{i-1}, -\frac{1}{h} u_{i-1} \! + \! \frac{1}{h} \omega_{i-1},0,1)
$.
The set wherein discrete evolution (jumps) are allowed to take place is named jump set and it is defined as 
$
D_{\xi} \coloneqq \R^5 \times T_s \Theta_\Delta
$,
where $\Theta_\Delta \coloneqq \{1,2, \ldots,\Delta+1 \}$. 
According to the definition of such a jump set, the system \eqref{eq:firstModel} experiences jumps whenever $\sigma_{i-1}$ is equal to $T_s$, $2T_s$, $3T_3$, $\ldots$, $(\Delta+1)T_s$.
Instantaneous jumps follow the equation 
$
\xi_i^+ = g_{\xi}(\xi_i) := (e_{1,i}, e_{2,i}, e_{3,i}, u_{i-1}, u_{i-1}, 0)
$,
where $g_{\xi}$ is named jump map.

\begin{remark}
The model in \eqref{eq:firstModel} considers only successful transmissions. In particular, successful transmissions occur for $\sigma_{i-1}=T_s$ when no DOS occurs, whereas they occur for $\sigma_{i-1}\in\iota T_s$, with $\iota:=\{2, \ldots,\Delta+1 \}$, for DOS attacks generating a number of consecutive packet dropouts within $1$ and $\Delta$. This characteristic is captured by the definition of $D_{\xi}$. Furthermore, by definition, $C_\xi$ and $D_\xi$ overlap each other, and, when $\xi$ belongs to $C_\xi \cap D_\xi$, the state of the system can either flow or jump because of a successful transmission. As such, solutions to \eqref{eq:firstModel} are not unique. In this sense, our model captures all possible network behaviors in a unified fashion.
\end{remark}

At this stage, we introduce the change of coordinates $\eta_{i-1}:=\hat{u}_{i-1}-u_{i-1}$,
which leads, by straightforward calculation, to the following hybrid system:
\begin{equation} \label{eq:systemH}
\mathcal{H}_i \left\lbrace 
\begin{array}{ll}
\dot{x}_i = f(x_i,\omega_{i-1}) & x_i \in C, \omega_{i-1} \in \mathbb{R} \\
x_i^+ = g(x_i) & x_i \in D
\end{array}
\right. 
\end{equation}
where $x_i:=(\tilde{x}_i,\eta_{i-1},\sigma_{i-1}) \in \mathbb{R}^{6}$ is the state, and $\tilde{x}_i:=(e_{1,i},e_{2,i},e_{3,i},u_{i-1}) \in \mathbb{R}^{4}$. The flow map is given by
\begin{equation}\label{eq:flowMapHy}
f(x_i,\omega_{i-1}) :=  \left( f_{\tilde{x}}(\tilde{x}_i,\eta_{i-1},\omega_{i-1}),  f_{\eta}(\tilde{x}_i,\eta_{i-1},\omega_{i-1}),1 \right)
\end{equation}
where $f_{\tilde{x}}(\tilde{x}_i,\eta_{i-1},\omega_{i-1}) := A_{xx} \tilde{x}_i + A_{x\eta} \eta_{i-1} +  A_{x\omega} \omega_{i-1}$, $f_{\eta}(\tilde{x}_i,\eta_{i-1},\omega_{i-1}) := A_{\eta x} \tilde{x}_i - \frac{1}{h} \omega_{i-1}$,
\begin{equation}
\begin{aligned}
& A_{xx} = \left[ \begin{array}{cc}
A_e & 0 \\
0 & -\frac{1}{h} 
\end{array} \right], \, A_{x\eta} = \left[\begin{array}{cccc}
0 & 0 & -\frac{1}{\tau_d} & 0
\end{array} \right]^\top\\
& A_{x\omega} = \left[ \begin{array}{cccc}
0 & 0 & 0 & \frac{1}{h}
\end{array} \right]^\top, A_{\eta x} = \left[ \begin{array}{cccc}
0 & 0 & 0 & \frac{1}{h}
\end{array} \right]
\end{aligned}
\end{equation}
and
\begin{equation} \label{eq:AeMatrix}
A_e = \left[ \begin{array}{ccc}
0 & 1 & 0 \\
0 & 0 & 1 \\ 
- \frac{k_p}{\tau_d} & - \frac{k_d}{\tau_d} & -\frac{1}{\tau_d}
\end{array} \right]
\end{equation}
The jump map is given by $g(x_i) := (\tilde{x}, 0, 0)$, whereas the flow set and the jump set are respectively given by $C \coloneqq  \R^5 \times [0,(\Delta+1)T_s]$, and $D \coloneqq \R^5 \times T_s \Theta_\Delta$.

Similarly to \cite{dolk2017event}, we employ the signal $\omega_i$ as the performance output of the hybrid system $\mathcal{H}_i$ to evaluate string stability. In particular, we assume:
\begin{equation} \label{eq:perfOutput}
\omega_i:= C_\omega \tilde{x}_i + \eta_{i-1}
\end{equation}
where 
$C_\omega = \left[
\begin{array}{cccccc}
k_p & k_d & 0 & 1 
\end{array}
\right]$. 
By $\mathcal{H}_i^\omega$ we denote the hybrid system $\mathcal{H}_i$ augmented with the performance output $\omega_i$.

\subsection{Problem Formulation}
A platoon of vehicles controlled by CACC needs to accomplish two main goals\cite{dolk2017event}: 1) regulate the spacing error in \eqref{eq:spacingError}, and 2) attenuate disturbance and shock waves along the vehicle platooning, due, e.g., to speed variations of the leader vehicle. 

The first property is usually referred to as individual vehicle stability. When this is satisfied, if $\mathcal{V}_0$ travels at some constant speed, the CACC ensures that $\lim_{t\to\infty} e_i(t)=0$ for the rest of the vehicles in the platoon. Therefore, individual vehicle stability is strictly connected with the eigenvalues of $A_e$, which depend on gains $k_p$ and $k_d$. 
Moreover, the error dynamics reflect the dynamic response of the vehicle platooning and can influence passengers comfort.
To this end, performance requirements need to be taken into account along with the satisfaction of the individual vehicle stability. In this paper, we characterize performance requirements by introducing the set:
\begin{equation} \label{eq:performance}
\mathbb{P}(\lambda_{M},\zeta_m)\coloneqq\{  A\in \mathbb{R}^{n \times n} | \Lambda_{\max}(A) = \lambda_{M}, \, \zeta_{\min}(A)\geq \zeta_m \}
\end{equation}
where $\lambda_{M}<0$ and $\zeta_m\in(0,1]$ are design parameters. The set $\mathbb{P}(\lambda_{M},\zeta_m)$ defines constraints on the location of the eigenvalues of $A_e$ in the complex plane. 
In particular, the eigenvalues of $A_e$ associated with the slowest mode of the error dynamics have real part equal to $\lambda_M$, whereas the other eigenvalues have real part smaller than $\lambda_M$. In addition, if complex, the eigenvalues of $A_e$ also have damping ratio greater than $\zeta_m$. Graphically, the eigenvalues of $A_e$ are placed within the gray area in Fig. \ref{fig:complexplane} with the rightmost eigenvalues lying on the dashed segment.
To meet the required performance, we design controller gains such that $A_e\in\mathbb{P}(\lambda_{M},\zeta_m)$ given $\lambda_{M}<0$ and $\zeta_m\in(0,1]$.

\begin{figure}[thpb]
\centering
\includegraphics[scale=0.7]{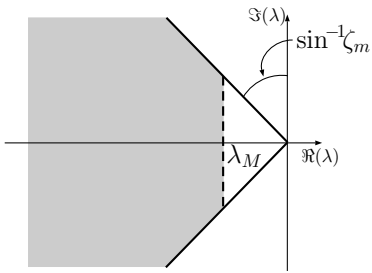}
\caption{Representation of the set $\mathbb{P}(\lambda_{M},\zeta_m)$ in the complex plane. To satisfy performance requirements, the proposed design procedure aims at placing the eigenvalues of $A_e$ within the gray area represented in this figure.}
\label{fig:complexplane}
\end{figure}

The second property we want to guarantee is also referred to as string stability of the vehicle platooning. It is related to the notion of input-output stability. String stability is widely investigated and analyzed via $\mathcal{L}_p$ stability \cite{ploeg2014lp}. 
Similarly to \cite{dolk2017event,oncu2012string,ploeg2014controller}, in this paper, we adopt the following notions to study string stability of vehicle platooning.  

\begin{definition}[$\tL2$ norm of a hybrid signal \cite{fichera2016lmi}]\label{def:tL2norm}
For a hybrid signal $\phi$, the $\tL2$ norm is given by
$
\Vert \phi \Vert_{\L2t} = \sqrt{\int_\mathcal{I} |\phi(r,j(r))|^2 dr}
$,
where $\mathcal{I}\coloneqq[0,\sup_t \dom\phi] \cap \dom_t \phi$. When $\Vert \phi \Vert_{\L2t}$ is finite, we say that $\phi\in\L2t$.
\end{definition}

\begin{definition}[$\L2t$-stability]\label{def:L2tstability}
The hybrid system $\mathcal{H}_i^\omega$ is said to be $\mathcal{L}_{2t}$-stable with respect to a closed set $\mathcal{A}$ from the input $\omega_{i-1} \in \L2t$ to the output $\omega_i \in \L2t$ with an $\mathcal{L}_{2t}$-gain less than or equal to $\theta$, if there exists $\alpha>0$ such that any solution pair\footnote{A pair $(\phi_i,\omega_{i-1})$ is a solution pair to $\mathcal{H}_i^{\omega}$ if it satisfies its dynamics; see \cite{cai2009characterizations} for more details.} $(\phi_i,\omega_{i-1})$ to $\mathcal{H}_i^{\omega}$ satisfies 
\begin{equation} \label{eq:L2tstability}
\Vert \omega_i \Vert_{\L2t}  \leq \alpha |\phi_i(0,0)|_\mathcal{A} + \theta \Vert \omega_{i-1}\Vert_{\L2t}
\end{equation}
\end{definition}

\begin{definition}[String stability]\label{def:Stringstability}
The vehicle platooning given by \eqref{eq:leader}, \eqref{eq:vehicleDynamics}, \eqref{eq:ZOH}, and \eqref{eq:controller} is said to be string stable if the hybrid systems $\mathcal{H}_i^\omega$, $i \in P_m$, are $\mathcal{L}_{2t}$-stable from the input $\omega_{i-1} \in \L2t$ to the output $\omega_i \in \L2t$ with an $\mathcal{L}_{2t}$-gain less than or equal to one.
\end{definition}

\begin{remark}
Because of vehicles homogeneity, the analysis of string stability of the whole vehicle platooning can be streamlined by focusing on the $\L2t$-stability of $\mathcal{H}_i^\omega$, for $i \in P_m$. A similar approach can be found in \cite{dolk2017event}.
\end{remark}

A critical aspect of the string stability is that it is negatively influenced by the IVC. Indeed, string stability can be degraded or compromised when network imperfections lead to prolonged unavailability of updated measurements; see  \cite{oncu2012string,kester2014critical,oncu2014cooperative}. Therefore, it is important to design a CACC such that the vehicle platooning maintains string stability for the largest achievable value of MANSD (identified by $\Delta$). Furthermore, estimating $\Delta$ provides an important metric for the evaluation of the resiliency of the overall platooning concerning the DOS attacks.

The problem we solve in this paper is formalized as follows:

\begin{problem} \label{prob:problem}
Given the platooning parameters $h$ and $\tau_d$, and $\mathbb{P}$ as in \eqref{eq:performance}, design gains $k_p$ and $k_d$ for the hybrid controller $\mathcal{K}$ such that the vehicle platooning given by \eqref{eq:leader}, \eqref{eq:vehicleDynamics}, \eqref{eq:ZOH}, and \eqref{eq:controller} satisfies the following properties with the largest achievable value of $\Delta$: 
\begin{enumerate}[({P}1)]
\item Individual vehicle stability with performance $\mathbb{P}$, i.e., $A_e\in\mathbb{P}$.
\item String stability, i.e., $\mathcal{L}_{2t}$-stability with an $\mathcal{L}_{2t}$-gain less than or equal to one, and $0$-input exponential stability.
\end{enumerate} 
\end{problem}

\section{Controller Design}
In this section, we illustrate our approach to solve Problem~\ref{prob:problem}. After describing how to meet the required performance for the individual vehicle stability, we provide sufficient conditions for $\L2t$-stability of $\mathcal{H}_i^\omega$, i.e., string stability of the vehicle platoon. Finally, we describe a procedure to design the controller gains to maximize the value of $\Delta$, while satisfying the two stability properties.

\subsection{Individual Vehicle Stability with Performance Requirements} \label{sec:indStability} 
To ensure individual vehicle stability with satisfactorily performance, we design $k_p$ and $k_d$ such that $A_e\in\mathbb{P}(\lambda_{M},\zeta_m)$. In particular, we aim at identifying values of $k_p$ and $k_d$ such that for any matrix $A_e\in\mathbb{P}(\lambda_{M},\zeta_m)$ one of the following conditions holds:
\begin{enumerate}[({C}1)]
\item $A_e$ has a unique real eigenvalue equal to $\lambda_{M}$ and two complex conjugate eigenvalues with real part less than or equal to $\lambda_{M}$ with damping ratio greater than $\zeta_m$;  \label{C1}
\item $A_e$ has a single couple of complex conjugate eigenvalues with real part equal to $\lambda_{M}$ and damping ratio greater than $\zeta_m$, and the other real eigenvalue is less than $\lambda_{M}$. \label{C2}
\end{enumerate}
\begin{remark}
Notice that, due to $A_e\in \R^{3\times3}$, $A_e\in\mathbb{P}$ if and only if either C\ref{C1} or C\ref{C2} are satisfied. In particular, to satisfy $A_e\in\mathbb{P}$ either C\ref{C1} or C\ref{C2} must hold. To this end, notice that conditions C\ref{C1} and C\ref{C2} can hold simultaneously for some specific selection of $k_p$ and $k_d$. When this happens, all the eigenvalues of $A_e$ have the same real part, which is equal to $\lambda_M$. 
\end{remark}
Next, we provide necessary and sufficient conditions on $k_p$ and $k_d$ such that C\ref{C1} or C\ref{C2} hold.

\begin{proposition}[N.S.C. for C1] \label{prop:C1}
Let $k_p, \, k_d \in \R$, $\lambda_{M}\in \R_{<0}$, and $\zeta_m\in \R_{>0}$. Then, C\ref{C1} is satisfied if and only if the following conditions hold:

\begin{subequations}
\label{eq:Condition1} 
\begin{align} 
&	\begin{aligned}
		k_d = f_{C1}(k_p):= -\frac{1}{\lambda_M} k_p - \lambda_M^2 \tau_d-\lambda_M \label{P1:kpf} \\
	\end{aligned}\\
&	\begin{aligned}
		k_p \leq \frac{|\lambda_M|(\lambda_M \tau_d + 1)^2}{4 \tau_d \zeta_m^2} \coloneqq \bar{k}_{p_{C1}} \label{P1:zeta}\\
	\end{aligned}\\
&	\begin{aligned}
		k_p \geq 2 \tau_d \lambda_M^3 +\lambda_M^2 \coloneqq \underaccent{\bar}{k}_{p_{C1}} \label{P1:others} \\
    \end{aligned} \\
&	\begin{aligned}
    	\lambda_M > - \frac{1}{3\tau_d} \label{P1:lamda} 
    \end{aligned}
%     &	\begin{aligned}
%		\scalebox{1}{$k_d = f_{C1}(k_p):= -\frac{1}{\lambda_M} k_p - \lambda_M(\lambda_M \tau_d+1)$} \label{P1:kpf} \\
%	\end{aligned}\\
%&	\begin{aligned}
%		\scalebox{1}{$k_p \leq \frac{|\lambda_M|(\lambda_M \tau_d + 1)^2}{4 \tau_d \zeta_m^2} \coloneqq \bar{k}_{p_{C1}}$} \label{P1:zeta}\\
%	\end{aligned}\\
%&	\begin{aligned}
%		\scalebox{1}{$k_p \geq 2 \tau_d \lambda_M^3 +\lambda_M^2 \coloneqq \underaccent{\bar}{k}_{p_{C1}}$} \label{P1:others} \\
%    \end{aligned} \\
%&	\begin{aligned}
%    	\scalebox{1}{$\lambda_M > - \frac{1}{3\tau_d}$} \label{P1:lamda} 
%    \end{aligned}
\end{align}
\end{subequations}
\end{proposition}
\medskip

\begin{proof}
\textit{Sufficiency}: Assume that \eqref{eq:Condition1} hold and let 
\begin{equation}
\scalebox{0.95}{$\rho(s;k_p,k_d):= \det(sI-A_e) = \frac{k_p}{\tau_d}+\frac{k_d}{\tau_d} s+ \frac{1}{\tau_d}s^2 + s^3$}
\end{equation}
be the characteristic polynomial of $A_e$. By replacing the expression of $k_d$ in \eqref{P1:kpf}, one gets
\scalebox{0.9}{$
\rho(s;k_p):= \rho_1(s; k_p)(s-\lambda_M)
$}
where 
\begin{equation}
\scalebox{0.95}{$\rho_1(s; k_p) :=s^2+\frac{\lambda_M \tau_d+1}{\tau_d} s-\frac{k_p}{\lambda_M \tau_d}$}
\end{equation}
This shows that $\lambda_M$ is an eigenvalue of $A_e$. Now, observe that 
\begin{equation}
\scalebox{0.95}{$
\zeta(\rho_1(s; k_p))=(\lambda_M \tau_d+1)/(2 \tau_d \sqrt{-\frac{k_p}{\lambda_M \tau_d}}) 
$}
\end{equation}
hence, from \eqref{P1:zeta} it follows that 
\scalebox{0.9}{$
\zeta(\rho_1(s; k_p))\geq\zeta_{m}
$}.
To conclude, it suffices to observe that, thanks to the Routh-Hurwitz criterion, \eqref{P1:others} and \eqref{P1:lamda} ensure that the real part of the eigenvalues of $\rho_1(s; k_p)$ is less than or equal to $\lambda_M$. Hence, \eqref{eq:Condition1} implies C1. This concludes the proof of sufficiency. 

\textit{Necessity}: Assume that C1 holds. Then, it follows that $\rho(s;k_p, k_d)$ can be factorized as follows
\scalebox{0.9}{$
\rho(s;k_p, k_d)=(s-\lambda_M)\rho_1(s;k_p, k_d)
$},
and $\rho_1$ is such that \scalebox{0.9}{$\zeta(\rho_1(s;k_p, k_d))\geq\zeta_{m}$} and \scalebox{0.9}{$\Re(\rho_1(s;k_p, k_d))\leq\lambda_M$}. Straightforward calculations yields
\begin{equation}
\scalebox{0.9}{$
\rho_1(s; k_p, k_d)=s^2+\frac{\lambda_M \tau_d+1}{\tau_d} s+\frac{\lambda_M^2 \tau_d+\lambda_M+k_d}{\tau_d}
$}
\end{equation}
which, in turn, shows that \scalebox{0.9}{$\zeta(\rho_1(s;k_p, k_d))\geq\zeta_{m}$} and \scalebox{0.9}{$\Lambda_{\max}(\rho_1(s;k_p, k_d))\leq\lambda_M$} implies \eqref{eq:Condition1}. This concludes the proof of necessity.
\end{proof}

Now we provide necessary and sufficient conditions on the gains $k_p$ and $k_d$ to guarantee C\ref{C2}.
\begin{proposition}[N.S.C. for C2] \label{prop:C2}
Let $k_p, \, k_d \in \R$, $\lambda_{M}\in \R_{<0}$, and $\zeta_m\in (0,1)$. Then, C\ref{C2} holds if and only if the following conditions hold:
\begin{subequations}
\label{eq:Condition2}
\begin{align}
 &	\begin{aligned}
    	k_d = f_{C2}(k_p):= - \frac{8 \lambda_M^3 \tau_d^2 + 8 \lambda_M^2 \tau_d +2 \lambda_M -\tau_d k_p }{2 \lambda_M \tau_d +1} \label{P2:kpf} \\
	\end{aligned}\\    
 &	\begin{aligned}
    	k_p \leq \frac{\lambda_M^2 (2 \lambda_M \tau_d+1)}{\zeta_m^2} \coloneqq \bar{k}_{p_{C2}} \label{P2:zeta}\\
	\end{aligned}\\    
 &	\begin{aligned}    
    	k_p >  2 \tau_d \lambda_M^3 +\lambda_M^2  \coloneqq \underaccent{\bar}{k}_{p_{C2}} \label{P2:others} \\
	\end{aligned}\\    
 &	\begin{aligned}    
    	\lambda_M > - \frac{1}{3\tau_d} \label{P2:lamda}      
   	\end{aligned}
%   	 &	\begin{aligned}
%    	\scalebox{1}{$k_d = f_{C2}(k_p):= - \frac{8 \lambda_M^3 \tau_d^2 + 8 \lambda_M^2 \tau_d +2 \lambda_M -\tau_d k_p }{2 \lambda_M \tau_d +1}$} \label{P2:kpf} \\
%	\end{aligned}\\    
% &	\begin{aligned}
%    	\scalebox{1}{$k_p \leq \frac{\lambda_M^2 (2 \lambda_M \tau_d+1)}{\zeta_m^2} \coloneqq \bar{k}_{p_{C2}}$} \label{P2:zeta}\\
%	\end{aligned}\\    
% &	\begin{aligned}    
%    	\scalebox{1}{$k_p >  2 \tau_d \lambda_M^3 +\lambda_M^2  \coloneqq \underaccent{\bar}{k}_{p_{C2}}$} \label{P2:others} \\
%	\end{aligned}\\    
% &	\begin{aligned}    
%    	\scalebox{1}{$\lambda_M > - \frac{1}{3\tau_d}$} \label{P2:lamda}      
%   	\end{aligned}
\end{align}
\end{subequations}
\end{proposition}
\medskip

\begin{proof}
\textit{Sufficiency}: Assume that \eqref{eq:Condition2} hold and let $\rho(s;k_p,k_d)$
be the characteristic polynomial of $A_e$. By using the expression of $k_d$ in \eqref{P2:kpf}, one gets
\begin{equation}
\scalebox{0.95}{$
\rho(s; k_p)\!\!=\!\!\underbrace{\left( s \! +\!\frac{2 \lambda_M \tau_d+1}{\tau_d} \right)}_{\rho_1(s; k_p)}\!\underbrace{\left(s^{2} \!- \! 2  \lambda_M s \!+\!\frac {k_p}{2 \lambda_M \tau_d+1}\right)}_{\rho_2(s; k_p)}
$}
\end{equation}
At this stage, notice that \eqref{P2:lamda} implies that the unique root of $\rho_1$ is smaller than $\lambda_M$. To conclude, we analyze the roots of $\rho_2$. Specifically, from the definition of $\rho_2$ it turns out that
\scalebox{0.9}{$
\zeta(\rho_2(s;k_p))=-\lambda_M/(\sqrt{\frac{k_p}{2 \lambda_M \tau_d+1}})
$}, 
which from \eqref{P2:zeta}-\eqref{P2:others} gives \scalebox{0.9}{$1>\zeta(\rho_2(s;k_p))\geq\zeta_m$}; this ensures that \scalebox{0.9}{$\Im(\rho_2(s;k_p))\neq 0$}. Moreover, straightforward calculations show that
\scalebox{0.9}{$
\Lambda_{\max}(\rho_2(s;k_p))=\lambda_M
$}.
Thus, C2 holds and this concludes the proof of sufficiency.

\textit{Necessity}: Assume that C2 holds. Then, it follows that $\rho(s;k_p, k_d)$ can be factorized as follows
\begin{equation}
\scalebox{0.95}{$
\rho(s;k_p, k_d)=(s-\lambda)\underbrace{(s-\lambda_M + j \omega)(s-\lambda_M - j \omega)}_{\rho_3(s;k_p,k_d)}
$}
\end{equation}
with \scalebox{0.9}{$\lambda \leq \lambda_M$}, \scalebox{0.9}{$\omega>0$}, and $\rho_3$ is such that \scalebox{0.9}{$1>\zeta(\rho_3(s;k_p,k_d))\geq\zeta_{m}$}. 
By solving the system of equations \scalebox{0.9}{$\Re(\rho(\lambda_M+ j \omega; k_p, k_d)) = 0$} and \scalebox{0.9}{$\Im(\rho(\lambda_M + j \omega; k_p, k_d)) = 0$} in the variables $k_p$ and $\omega$ one obtains \eqref{P2:kpf}.  By replacing \eqref{P2:kpf} in $\rho_3$ one gets \scalebox{0.9}{$\rho_3(s;k_p) := s^2-2 s \lambda_M + \frac{k_p}{2 \lambda_M \tau_d+1}$}. At this stage, straightforward calculations yield \scalebox{0.9}{$1 > \zeta(\rho_3(s;k_p))\geq\zeta_{m}$} and \scalebox{0.9}{$\Lambda_{\max}(\rho_3(s;k_p))=\lambda_M$}, which, in turn, implies, respectively, \eqref{P2:zeta} and \eqref{P2:others}.  Moreover, from \scalebox{0.9}{$\rho(s;k_p, f_{C2}(k_p))$}, one has that \scalebox{0.9}{$\lambda \coloneqq -\frac{2 \lambda_M \tau_d+1}{\tau_d}$} and since \scalebox{0.9}{$\lambda \leq \lambda_M$}, \eqref{P2:lamda} holds. This concludes the proof of necessity.
\end{proof}

\subsection{Sufficient Conditions for Platooning Stability} \label{sec:stability} 
The previous subsection describes how to obtain the set of gains $k_p$ and $k_d$ such that the individual vehicle stability has satisfactory performance. In this subsection, we consider $k_p$ and $k_d$ as given, and we study the stability of $\mathcal{H}_i^\omega$, which also includes the network dynamics.

Our approach aims at formulating the control problem as a set stabilization problem. In particular, our approach consists of analyzing the stability properties of the following compact set
\begin{equation} \label{eq:setA}
\mathcal{A} \coloneqq \{0\} \times \{0\} \times [0,(\Delta+1)T_s]
\end{equation}
The following definition formalizes these properties.
\begin{definition}[Exponential input-to-state stability \cite{MercoLCSS}]
\label{defISS}
Let $\mathcal{A}\subset\R^{6}$ be closed. The hybrid system $\mathcal{H}_i^{\omega}$ is \emph{exponentially input-to-state-stable} (eISS) with respect to the set $\mathcal{A}$ if there exist $\kappa, \lambda>0$, and $p\in\mathcal{K}$ such that each maximal solution pair $(\phi_i,\omega_{i-1})$ to $\mathcal{H}_i^{\omega}$ is complete, and, if $\Vert \omega_{i-1} \Vert_\infty$ is finite, it satisfies
\begin{equation}
\label{eq:defeISS}
\vert \phi_i(t,j)\vert_{\mathcal{A}}\leq \max\{\kappa e^{-\lambda (t+j)}\vert \phi_i(0,0)\vert_{\mathcal{A}}, p(\Vert \omega_{i-1} \Vert_{\infty})\}
\end{equation}
for each $(t,j)\in\dom\phi_i$, where $\Vert \omega_{i-1} \Vert_\infty$ denotes the $\mathcal{L}_\infty$ norm of the hybrid signal $\omega_{i-1}$ as defined in \cite{nevsic2013finite}.
\end{definition} 
Moreover, to satisfy string stability, $\mathcal{H}_i^{\omega}$, $i \in P_m$, must be $\L2t$-stable from the input $\omega_{i-1}$ to the output $\omega_i$ with an $\mathcal{L}_{2t}$-gain less than or equal to one.
It is worth mentioning that whenever eISS and $\L2t$-stability are satisfied, the vehicle platooning given by \eqref{eq:leader}, \eqref{eq:vehicleDynamics}, \eqref{eq:ZOH}, and \eqref{eq:controller} satisfies individual vehicle stability and string stability. In the following, we identify sufficient conditions to ensure those two stability properties. 
First, we employ Lyapunov theory for hybrid systems to provide conditions for eISS and $\L2t$-stability of $\mathcal{H}_i^{\omega}$ (Assumption~\ref{ass:lyapunov} and  Theorem~\ref{th:suffCond}). Then, we give sufficient conditions for eISS and $\L2t$-stability in the form of matrix inequalities (Theorem~\ref{th:practical} and Lemma~\ref{lemma:Mco}).

Consider the following assumption. 

\begin{assumption} \label{ass:lyapunov}
There exist two continuously differentiable functions
$V_1 : \mathbb{R}^{4} \rightarrow  \mathbb{R}$, $V_2 : \mathbb{R}^{2} \rightarrow  \mathbb{R}$, and positive real numbers $\alpha_1$, $\alpha_2$, $\beta_1$, $\beta_2$, $\lambda_t$, and $\epsilon$ such that
\begin{enumerate}[({A}1)]
\item $\alpha_1 |\tilde{x}_i|^2 \leq V_1(\tilde{x}_i) \leq \alpha_2 |\tilde{x}_i|^2, \quad \forall x_i \in C$ \label{ass:lyapunov:A1}
\item $\beta_1 |\eta_{i-1}|^2 \leq V_2(\eta_{i-1},\sigma_{i-1}) \leq \beta_2 |\eta_{i-1}|^2, \quad \forall x_i \in C$
\label{ass:lyapunov:A2}
\item $V_2(0,0) \leq V_2(\eta_{i-1},\sigma_{i-1})$, $\forall \eta_{i-1} \in \R, \, \sigma_{i-1} \in T_s \Theta_\Delta$ \label{ass1:jump}
\item the function $x_i \mapsto V(x_i) := V_1(\tilde{x}_i)+V_2(\eta_{i-1},\sigma_{i-1})$ satisfies  $\langle \nabla V(x_i), f(x_i,\omega_{i-1}) \rangle \leq -2\lambda_t V(x_i) - \omega_i^2 + \theta^2 \omega_{i-1}^2$ for each $x_i \in C, \omega_{i-1} \in \mathbb{R}$, where $\omega_i^2 = \tilde{x}_i^\top C_\omega^\top C_\omega \tilde{x}_i + \eta_{i-1}^2 + 2 C_\omega \tilde{x}_i \eta_{i-1}$ from \eqref{eq:perfOutput}. \label{ass1:Vdot}
\end{enumerate} 
\end{assumption}

Based on Assumption~\ref{ass:lyapunov}, the result given next provides sufficient conditions for eISS and $\L2t$-stability of $\mathcal{H}_i^\omega$. 

\begin{theorem} \label{th:suffCond}
Let Assumption \ref{ass:lyapunov} hold. Then:
\begin{enumerate}[($i$)]
\item The hybrid system $\mathcal{H}_i^{\omega}$ is eISS with respect to $\mathcal{A}$;
\item The hybrid system $\mathcal{H}_i^{\omega}$ is $\L2t$-stable from input $\omega_{i-1}$ to output $\omega_i$ with an $\mathcal{L}_{2t}$-gain less than or equal to $\theta$.
\end{enumerate}
\end{theorem}

\begin{proof}
Inspired by \cite{ferrante2015hybrid}, we select, for every $x_i \in \mathbb{R}^{6}$, $V(x_i) := V_1(\tilde{x}_i)+V_2(\eta_{i-1},\sigma_{i-1})$ as a Lyapunov function candidate for the hybrid system $\mathcal{H}_i^{\omega}$. We prove ($i$) first.
Select $\rho_1 = \min\{\alpha_1,\beta_1\}$, $\rho_2 = \max\{\alpha_2,\beta_2\}$. By considering the definition of the set $\mathcal{A}$ in \eqref{eq:setA}, one obtains
\begin{equation} \label{eq:proofTheoremLyapBounds}
\rho_1 |x_i|_\mathcal{A}^2 \leq V(x_i) \leq \rho_2 |x_i|_\mathcal{A}^2 \qquad \forall x_i \in C \, \cup \,  D
\end{equation}
Moreover, from Assumption \ref{ass:lyapunov} item (A\ref{ass1:Vdot}) one has that $\forall x_i \in C,\omega_{i-1} \in \mathbb{R}$
\begin{equation} \label{eq:proofTheoremLyap}
\langle \nabla V(x_i), f(x_i,\omega_{i-1}) \rangle \leq -2\lambda_t V(x_i) + \theta^2 \omega_{i-1}^2
\end{equation}
and from Assumption \ref{ass:lyapunov} item (A\ref{ass1:jump}), one has that for all $x_i \in D$
\begin{equation} \label{eq:proofJumps}
V(g(x_i))\leq V(x_i)
\end{equation}
Let $(\phi_i,\omega_{i-1})$ be a maximal solution pair to $\mathcal{H}_i^{\omega}$. 
By using \eqref{eq:proofTheoremLyapBounds}, \eqref{eq:proofTheoremLyap}, and \eqref{eq:proofJumps} and following the same steps as in \cite[proof of Theorem 1]{ferrante2015hybrid}, one obtains that for all $(t,j)\in \textrm{dom}\phi_i$ 
\begin{equation}\label{eq:proofTheoremRate}
|\phi_i(t,j)|_\mathcal{A} \! \leq \! \textrm{max} \! \left\lbrace \! 2\sqrt{\frac{\rho_2}{\rho_1}}e^{-\lambda_t t} |\phi_i(0,\!0)|_\mathcal{A} ,\!  \frac{2 \theta}{\sqrt{2 \lambda_t \rho_1}} \Vert\omega_{i-1}\Vert_\infty \!\! \right\rbrace
\end{equation}
which reads as \eqref{eq:defeISS} with $\kappa = 2\sqrt{\frac{\rho_2}{\rho_1}}$, $\lambda = \lambda_t$ and $r \mapsto p(r) := (2\theta/\sqrt{2 \lambda_t \rho_1}) r$. Hence, since every maximal solution pair to $\mathcal{H}_i^{\omega}$ is complete, ($i$) is established.
Now we prove ($ii$). Let $(\phi_i, \omega_{i-1})$ be a maximal solution pair to
$\mathcal{H}_i^{\omega}$ and select $t>0$. Notice that because of Assumption \ref{ass:lyapunov} item (A\ref{ass1:jump}) $V$ is nonincreasing at jumps. Similarly to \cite[proof of Theorem~1]{ ferrante2015hybrid}, one obtains {\small $\int_{\mathcal{I}(t)} \omega_i(r,j(r))^2 dr  \leq  V(\phi_i(0,0)) + \theta^2 \int_{\mathcal{I}(t)} \omega_{i-1}(r,j(r))^2 dr$} where ${\mathcal{I}(t)} := [0,t] \cap \textrm{dom}_t \phi_i$ and $\textrm{dom}_t \phi_i \coloneqq \{t \in \mathbb{R}_{\geq 0} : \exists j \in \mathbb{N}_0 \, \text{s.t.} \, (t, j) \in \textrm{dom} \phi_i\}$.
%\footnote{Given a hybrid signal $w$, then $\textrm{dom}_t w \coloneqq \{t \in \mathbb{R}_{\geq 0} : \exists j \in \mathbb{N}_0 \, \text{s.t.} \, (t, j) \in \textrm{dom} \,w\}$. See \cite{ferrante2015hybrid} for further details.}.
To conclude, one can take the limit for $t$ approaching $\textrm{sup}_t \, \textrm{dom}\phi_i$, and by considering \eqref{eq:proofTheoremLyapBounds}, one obtains \eqref{eq:L2tstability} with $\alpha=\rho_2$. Hence, the result ($ii$) is established.
\end{proof}

\begin{theorem} \label{th:practical}
Let parameters $\tau_d$, $h$, $k_p$, $k_d$, the transmission interval $T_s$, the MANSD $\Delta$, and $\theta\in\R$ be given. If there exist $P_1 \in \mathcal{S}^{4}_+$, $p_2\in \R_{>0}$,  and $\delta \in \R_{>0}$ such that
\begin{equation}
\label{eq:MnegAllt}
\mathcal{M}(\sigma_{i-1}) < 0, \quad \forall \sigma_{i-1} \in [0,(\Delta+1)T_s]
\end{equation}
where the function $[0,(\Delta+1)T_s] \ni \sigma_{i-1} \mapsto \mathcal{M}(\sigma_{i-1})$ is given by
\begin{equation}
\label{LMI} 
\resizebox{\hsize}{!}{$
\begin{aligned}
& \mathcal{M}(\sigma_{i-1})= \\
& \left[ 
\begin{array}{ccc}
\textrm{He}(P_1 A_{xx})+ C_\omega^\top C_\omega &  P_1 A_{x\eta} + C_\omega^\top +   e^{-\delta \sigma_{i-1}} p_2 A_{\eta x}^\top  & P_1 A_{x\omega}\\
\bullet &  - \delta p_2 e^{-\delta \sigma_{i-1}} + 1 &  - e^{-\delta \sigma_{i-1}}  p_2 / h \\
%\bullet & \bullet  & -(1+\epsilon)\\
\bullet & \bullet  & -\theta^2\\
\end{array}
\right]
\end{aligned}
$}
\end{equation}
Then, functions $\tilde{x}_i \mapsto V_1(\tilde{x}_i) \coloneqq \tilde{x}_i^{\top}P_1 \tilde{x}_i$, and $(\eta_{i-1},\sigma_{i-1}) \mapsto V_2(\eta_{i-1},\sigma_{i-1}) \coloneqq p_2 \eta_{i-1}^2 e^{-\delta \sigma_{i-1}}$ satisfy Assumption~\ref{ass:lyapunov}. 
\end{theorem}

\begin{proof}
Consider the functions $\tilde{x}_i \mapsto V_1(\tilde{x}_i) \coloneqq \tilde{x}_i^{\top}P_1 \tilde{x}_i$, and $(\eta_{i-1},\sigma_{i-1}) \mapsto V_2(\eta_{i-1},\sigma_{i-1}) \coloneqq p_2 \eta_{i-1}^2 e^{-\delta \sigma_{i-1}}$. By choosing $\alpha_1=\lambda_{\min}(P_1)$, $\alpha_2=\lambda_{\max}(P_1)$, $\beta_1=p_2 e^{-\delta (\Delta+1)T_s}$, $\beta_2=p_2$, it turns out that items (A\ref{ass:lyapunov:A1}) and (A\ref{ass:lyapunov:A2}) of the Assumption \ref{ass:lyapunov} hold. To show that item (A\ref{ass1:jump}) holds, notice that, by employing the jump map of $\mathcal{H}_i^\omega$, for all $\eta_{i-1} \in \R$ and for all $\sigma_{i-1} \in  T_s \Theta_\Delta$ one has that 
$
V_2(0,0) - V_2(\eta,\sigma_{i-1}) = - p_2 \eta_{i-1}^2 e^{-\delta \sigma_{i-1}} \leq 0
$.
Regarding item (A\ref{ass1:Vdot}) of Assumption~\ref{ass:lyapunov}, let $V(x_i)=V_1(\tilde{x}_i)+V_2(\eta_{i-1},\sigma_{i-1})$. Then, from the definition of the flow map in \eqref{eq:flowMapHy}, for each $x_i\in C$, $\omega_{i-1} \in \R$ one can define $\Omega(x_i,\omega_{i-1})\coloneqq\langle \nabla V(x_i), f(x_i,\omega_{u-1}) \rangle+ \tilde{x}_i^\top C_\omega^\top C_\omega \tilde{x}_i + \eta_{i-1}^2 + 2 C_\omega \tilde{x}_i \eta_{i-1}  - \theta^2\omega_{i-1}^2$. Therefore, by defining $\Psi(x_i,\omega_{i-1}):=(\tilde{x}_i,\eta_{i-1},\omega_{i-1})$, for each $x_i \in C$ and $\omega_{i-1} \in \R$, one has $\Omega(x_i,\omega_{i-1})  \coloneqq \Psi(x_i,\omega_{i-1})^\top \mathcal{M}(\sigma_{i-1}) \Psi(x_i,\omega_{i-1})$ where the symmetric matrix $\mathcal{M}(\sigma_{i-1})$ is given in \eqref{LMI}.  The satisfaction of $\mathcal{M}(\sigma_{i-1})<0, \, \forall \sigma_{i-1} \in [0,(\Delta+1)T_s]$ leads to $\overline{\varsigma}\coloneqq \max_{\sigma_{i-1}\in[0,(\Delta+1)T_s]}\lambda_{\max}(\mathcal{M}(\sigma_{i-1}))<0$. Observe that, since $\sigma_{i-1} \mapsto \mathcal{M}(\sigma_{i-1})$ is continuous on $[0,(\Delta+1)T_s]$, $\overline{\varsigma}$ is well defined. Therefore, one has that for all $x_i\in C, \, \omega_{i-1}\in\R$, $\Omega(x_i,\omega_{i-1})\leq -\overline{\varsigma} \tilde{x}_i^\top \tilde{x}_i=-\overline{\varsigma}\vert x_i\vert^2_\mathcal{A}$.
To conclude, let $\rho_2 \coloneqq \max\{\alpha_2,\beta_2\}$, using \eqref{eq:proofTheoremLyapBounds} and the definition of $\Omega$, one has that for all $x_i\in C,\, \omega_{i-1}\in\R$, $\langle \nabla V(x_i), f(x_i,\omega_{i-1}) \rangle\leq -\frac{\overline{\varsigma}}{\rho_2}V(x_i) - \tilde{x}_i^\top C_\omega^\top C_\omega \tilde{x}_i - \eta_{i-1}^2 - 2 C_\omega \tilde{x}_i \eta_{i-1}+ \theta^2 \omega_{i-1}^2$ which reads as (A\ref{ass1:Vdot}). Hence, item (A\ref{ass1:Vdot}) holds. This concludes the proof.
\end{proof}

The following lemma is employed to reduce the complexity in the use of Theorem~\ref{th:practical}. Indeed, it allows to convert the infinite set of matrix inequalities in 
\eqref{eq:MnegAllt} to only two matrix inequalities in \eqref{eq:Mconditions}. The proof of this result follows the same steps as in  \cite{ferrante2015hybrid} and it is omitted.

\begin{lemma} \label{lemma:Mco}
Let $P_1 \in \mathcal{S}^{4}_+$, $p_2$, $\delta$, $\tau_d$, $h$, and $T_s$ be given positive real number, $\Delta\in\N_0$, and $k_p$, and $k_d$ be given real numbers. For each $\sigma_{i-1} \in [0,(\Delta+1)T_s]$, define $\mathcal{M}: \sigma_{i-1} \mapsto \mathcal{M}(\sigma_{i-1})$. Then, $\text{rge} \mathcal{M} = \text{Co}\{\mathcal{M}(0),\mathcal{M}((\Delta+1)T_s)\}$. Therefore, \eqref{eq:MnegAllt} holds if and only if 
\begin{equation}
\label{eq:Mconditions}
\begin{array}{cc}
\mathcal{M}(0)<0, & \mathcal{M}((\Delta+1)T_s)<0
\end{array}
\end{equation}
\end{lemma}

The satisfaction of Theorem~\ref{th:practical} leads to stability of $\mathcal{H}_i^\omega$ with $\L2t$-gain less than or equal to $\theta$. Notice that the requirement for string stability is $\L2t$-gain less than or equal to one. However, similarly to \cite{dolk2017event}, to make condition \eqref{eq:Mconditions} feasible, we consider an $\L2t$-gain less than or equal to $\theta=\sqrt{1+\epsilon}$ with $\epsilon$ is a small strictly positive value.

\subsection{Controller Tuning Algorithm}
In the following, we show how to employ Proposition~\ref{prop:C1}, Proposition~\ref{prop:C2}, and Theorem~\ref{th:practical} to devise a procedure for the selection of gains $k_p$, and $k_d$ able to solve Problem~\ref{prob:problem}.

Employing Theorem~\ref{th:practical} and Lemma~\ref{lemma:Mco}, one can reformulate Problem~\ref{prob:problem} as the following optimization problem:
\begin{equation} \label{eq:optimProbl}
\underset{P_1,p_2,\delta, k_p, k_d}{\text{maximize}} \quad \Delta, \quad \text{subject to} \, A_e\in\mathbb{P}, \eqref{eq:Mconditions}
\end{equation} 
Notice that the optimization problem \eqref{eq:optimProbl} is nonlinear in the decision variables. For this reason, the solution to \eqref{eq:optimProbl} is difficult from a numerical point of view \cite{boyd1994linear}. In particular, notice that sufficient conditions \eqref{eq:Mconditions} are in the form of matrix inequalities that are nonlinear in $P_1$, $p_2$, $\delta$, $k_p$, and $k_d$. Therefore, they cannot be directly used as a computationally tractable design tool. 
On the other end, when $\delta$, $\Delta$, $k_p$, and $k_d$ are fixed, \eqref{eq:Mconditions} is linear in variables $P_1$ and $p_2$, hence, \eqref{eq:optimProbl} becomes a semidefinite program and can be solved by using available solvers.

Our proposed strategy to obtain a suboptimal solution to \eqref{eq:optimProbl} consists of operating a two-stage line search for the scalars $k_p$, $k_d$, $\delta$, and $\Delta$. The first stage consists of choosing values $(k_p,k_d)$ such that $A_e\in\mathbb{P}$. The second stage considers $(k_p,k_d)$ as given, and targets estimating the largest value of $\Delta$ by checking the feasibility of $\eqref{eq:Mconditions}$ through line searches for $\delta$ and $\Delta$. Observe that while a line searches for $\delta$ and $\Delta$ can be easily implemented with numerical algorithms, exploring values $(k_p,k_d)$ such that $A_e\in\mathbb{P}$ can be computationally expensive if $(k_p,k_d)$ are not suitably selected, e.g., by using gridding techniques on both $k_p$ and $k_d$. In this paper, we choose values $(k_p,k_d)$ by employing Propositions~\ref{prop:C1} and \ref{prop:C2}, which give upper and lower bounds on $k_p$ and yield to obtain $k_d$ as a function of $k_p$; see \eqref{eq:Condition1} and \eqref{eq:Condition2}. This results being one of the main contributions of this paper. In fact, by following this approach, $k_p$ becomes the only parameter for the first stage of the design algorithm, which employs only a bounded line search on $k_p$ with bounds known in advance. Hence, Propositions~\ref{prop:C1} and \ref{prop:C2} dramatically reduce the complexity of the design procedure.

To summarize, the design procedure we propose to solve Problem~\ref{prob:problem} is outlined in Algorithms~\ref{al:mainAlgorithm} and \ref{al:algorithmDelta}. In particular, Algorithm \ref{al:mainAlgorithm} provides the overall design procedure and calls Algorithm~\ref{al:algorithmDelta} for the second stage of the design strategy, i.e., estimating the value of $\Delta$ for given $(k_p,k_d)$.

In the following, we briefly analyze the computational complexity of Algorithms~\ref{al:mainAlgorithm} and \ref{al:algorithmDelta}. To this end, we employ the ``Big O" ($\mathcal{O}$) notation; see, e.g., \cite{cormen2001introduction}. 
Observe that execution time of Algorithm~\ref{al:mainAlgorithm} grows by increasing the size of vectors $\overrightarrow{k}_{p_{C1}}$ and $\overrightarrow{k}_{p_{C2}}$, whereas the execution time of Algorithm~\ref{al:algorithmDelta} grows by increasing the size of the line search on $\delta$. By employing a straightforward analysis of the worst-case iterations required by the design algorithm, one can conclude that Algorithm~\ref{al:algorithmDelta} has complexity $\mathcal{O}(n_\delta)$, whereas Algorithm~\ref{al:mainAlgorithm} has complexity $\mathcal{O}\left(n_\delta (n_{k_1}+n_{k_2})\right)$.

\begin{remark}
It is worth mentioning that the proposed design procedure is meant to run offline. The gains $k_p$ and $k_d$ obtained by the design algorithm are used afterward in the implemented control strategy. Therefore, the time required to obtain the optimal selection of gains $k_p$ and $k_d$ is not critical from a real-time implementation point of view. To this end, notice that the proposed controller is a PD with a ZOH. Its real-time implementation in existing vehicular embedded systems does not differ from a traditional PD controller.
\end{remark}

\begin{algorithm}[phtb]
\caption{Tuning algorithm for performance $\mathbb{P}$, string stability, and the largest achievable $\Delta$}\label{al:mainAlgorithm}
\hspace*{\algorithmicindent} \textbf{Input:} $T_s$, $\tau_d$, $h$
\algblockdefx[NAME]{StartOne}{EndOne}{\textbf{Step 1:} \textit{Explore $(k_p,k_d)$ such that $\Lambda(A_e)$ as in C1.}} {\textbf{End Step 1.}}
\algblockdefx[NAME]{StartTwo}{EndTwo}{\textbf{Step 2:} \textit{Explore $(k_p,k_d)$ such that $\Lambda(A_e)$ as in C2.}} {\textbf{End Step 2.}}
\begin{algorithmic}[1]
\StartOne
\State Define an array $\overrightarrow{k}_{p_{C1}}$ of dimension $n_{k_1}$ with elements in $[\underaccent{\bar}{k}_{p_{C1}}, \bar{k}_{p_{C1}}]$.
\State For each $k_p$ in $\overrightarrow{k}_{p_{C1}}$:
\State \hspace*{\algorithmicindent} Feed \textbf{Algorithm~\ref{al:algorithmDelta}} with $(k_p,f_{C1}(k_p))$.
\State \hspace*{\algorithmicindent} Store returned values of $\Delta$.
\State Identify $k_p$, $k_d$ such that $\Delta$ is the largest.
\EndOne
\StartTwo
\State Define an array $\overrightarrow{k}_{p_{C2}}$ of dimension $n_{k_2}$ with elements in $(\underaccent{\bar}{k}_{p_{C2}}, \bar{k}_{p_{C2}}]$.
\State For each $k_p$ in $\overrightarrow{k}_{p_{C2}}$:
\State \hspace*{\algorithmicindent} Feed \textbf{Algorithm~\ref{al:algorithmDelta}} with $(k_p,f_{C2}(k_p))$.
\State \hspace*{\algorithmicindent} Store returned values of $\Delta$.
\State Identify $k_p$, $k_d$ such that $\Delta$ is the largest.
\EndTwo
\State \textbf{Step 3:} Parameters of controller $\mathcal{K}$ are assigned with values of $k_p$, $k_d$ such that $\Delta$ is the largest among Step~1 and Step~2.
\end{algorithmic}
\end{algorithm}

\begin{algorithm}[phtb]
\caption{Given $(k_p,k_d)$, estimate the largest achievable $\Delta$}\label{al:algorithmDelta}
\hspace*{\algorithmicindent} \textbf{Input:} $T_s$, $\tau_d$, $h$, $k_p$, $k_d$
\begin{algorithmic}[1] 
\State Initialize $\Delta_i$, $\Delta$, $stop$ to zero.
\While{$stop=0$}
	\State line search on $\delta$ based on $n_\delta$ samples for given $\Delta_i$ such that \eqref{eq:Mconditions} is feasible.  
	\If{(line search succeeded)}
		\State $\Delta \gets \Delta_i$, $\Delta_i\gets \Delta_i+1$
	\Else 
	 	\State $stop \gets 1$
	\EndIf
\EndWhile
\State \Return $\Delta$
\end{algorithmic}
\end{algorithm}

\section{Numerical Results}
In this section, we apply Algorithm~\ref{al:mainAlgorithm} to tune the controller $\mathcal{K}$ for a homogeneous platooning of $11$ vehicles. In particular, we select performance $\mathbb{P}$ from \cite{ploeg2011design}, and we show the outcome of tuning the controller parameters $k_p$, $k_d$ by following the approach proposed in this paper. All numerical results are obtained by using Matlab$^{\tiny{\textregistered}}$.
%\footnote{Code at {\color{red} to be uploaded on github}}.
Semidefinite optimizations are performed by using \textit{YALMIP} \cite{lofberg2004yalmip} with solver \textit{SEDUMI} \cite{sturm1999using}. 

Numerical results are obtained by assuming a transmission rate for measurement $u_{i-1}$ equal to $20$ Hz ($T_s=0.05\, s$), as adopted in \cite{gao2016empirical}. Moreover, we select parameters $h = 0.7$, $\tau_d=0.1$, $\lambda_M=-0.367$ from \cite{ploeg2011design}, and $\zeta_m=0.7$. 
Let $\bar{\mathcal{K}}$ be the controller in \eqref{eq:controller} with gains $(k_p,k_d)=(0.2,0.7)$, as in \cite{ploeg2011design}, and consider $\hat{\mathcal{K}}$ the same controller tuned with our approach. The design of $\hat{\mathcal{K}}$ through Algorithm~\ref{al:mainAlgorithm} results in a final tuning characterized by $(k_p,k_d)=(0.82,2.6)$.  Observe that both $\bar{\mathcal{K}}$ and $\hat{\mathcal{K}}$ lead to vehicle platoons that satisfy performance $\mathbb{P}$ with $\lambda_M=-0.367$ and $\zeta_m=0.7$.

From a computational point of view, notice that the design employs $n_{k_1}=162$, $n_{k_2}=13$, and $\delta=241$ and requires a computation time of $1$ hour, $58$ minutes and $53$ seconds on a $2.70$ GHz Intel Core $i7$ RAM $32$ GB. 

\begin{figure}[htbp]
\centering 
\psfrag{x}[][][0.9]{$k_p$}
\psfrag{y}[b][][0.9]{$k_d$}
\includegraphics[clip, scale=0.45]{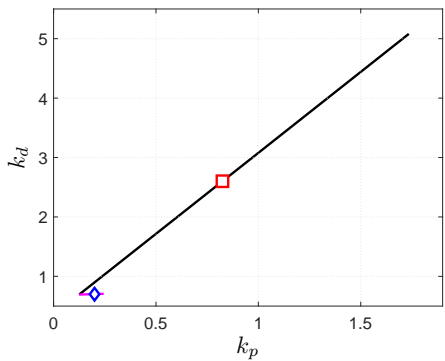} \\
\caption{Locus of $(k_p,k_d)$ such that Conditions C1 (black) and C2 (magenta) for $\Lambda(A_e)$ are satisfied. The blue diamond and the red square respectively identify $\Lambda(A_e)$ and $(k_p,k_d)$ for $\bar{\mathcal{K}}$ and $\hat{\mathcal{K}}$.}
\label{fig:setKopt}
\end{figure}

\begin{figure}[htpb]
\centering
\psfrag{x}[][][0.9]{$k_p$}
\psfrag{y}[b][][0.9]{$\Delta$}
\includegraphics[scale=0.45]{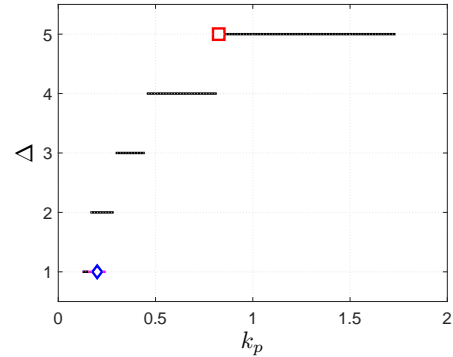}
\caption{Locus of $(k_p,\Delta)$ such that Conditions C1 (black) and C2 (magenta) for $\Lambda(A_e)$ are satisfied. The blue diamond and the red square respectively identify the final tuning $(k_p,\Delta)$ for $\bar{\mathcal{K}}$ and $\hat{\mathcal{K}}$.}
\label{fig:setDeltaopt}
\end{figure}

To better understand the approach, consider Fig. \ref{fig:setKopt} and Fig. \ref{fig:setDeltaopt}, which respectively represent the locus $(k_p,k_d)$ such that $A_e \in \mathbb{P}$, and the locus $(k_p,\Delta)$ such that $A_e \in \mathbb{P}$ and string stability are satisfied. Notice that the tuning of $\hat{\mathcal{K}}$ aims at selecting the minimum value of $k_d$ such that $\Delta$ is maximum. This choice allows reducing, at minimum, the effect of the derivative action on the controlled vehicles. By analyzing the tuning of controllers $\bar{\mathcal{K}}$ and $\hat{\mathcal{K}}$ in Fig.~\ref{fig:setDeltaopt}, one can conclude that $\hat{\mathcal{K}}$, tuned with Algorithm~\ref{al:mainAlgorithm}, guarantees performance $\mathbb{P}$ and string stability with higher resiliency to DOS attacks with respect to  $\bar{\mathcal{K}}$. This emerges from the fact that the value of $\Delta$ obtained for $\hat{\mathcal{K}}$ is equal to $5$, whereas $\Delta$ is equal to $1$ for $\bar{\mathcal{K}}$.

To validate our approach, we simulate a platoon of $11$ vehicles where the leader performs an acceleration of $2\,m/s^2$, and a deceleration of $4\,m/s^2$. This acceleration profile is similar to that one used in \cite{ploeg2014controller}. Figure~\ref{fig:speedTest_noDrops} and Fig.~\ref{fig:speedTest_Drops} depict velocity and distance profiles for vehicle platoons controlled by $\bar{\mathcal{K}}$ and $\hat{\mathcal{K}}$ respectively in case of ``attack-free" IVC and IVC affected by DOS attacks.
To show how the controlled platoons behave under DOS attacks, we consider the worst DOS attack case scenario: We induce DOS attacks with intervals characterized by $5$ consecutive packet dropouts and only $1$ packet successfully delivered in between DOS intervals. The simulations start with the first of the five packet dropouts, i.e., the first measurements exchanged by the vehicles is at $6T_s=0.3\,s$ after starting the simulations.
As a result, notice that in case of ``attack-free" IVC, the two controllers provide the same behavior. In case of occurring DOS attacks, instead, the behavior of the vehicle platooning controlled by $\bar{\mathcal{K}}$ is degraded compared to the same controller with ``attack-free" network and vehicle platoons controller by $\hat{\mathcal{K}}$ under DOS attacks. This is noticeable by the increase in overshoot for increasing vehicle index.

\begin{figure}[thpb]
\vspace{5pt}
\centering
\psfrag{x1}[][][0.9]{$t \, [s]$}
\psfrag{a}[][][0.9]{(a)}
\psfrag{y1}[][][0.9]{$v_i \, [m/s]$}
\psfrag{x2}[][][0.9]{$t \, [s]$}
\psfrag{b}[][][0.9]{(b)}
\psfrag{y2}[][][0.9]{$d_i \, [m]$}
\psfrag{x3}[][][0.9]{$t \, [s]$}
\psfrag{c}[][][0.9]{(c)}
\psfrag{y3}[][][0.9]{$v_i \, [m/s]$}
\psfrag{x4}[][][0.9]{$t \, [s]$}
\psfrag{d}[][][0.9]{(d)}
\psfrag{y4}[][][0.9]{$d_i \, [m]$}
\includegraphics[scale=0.6]{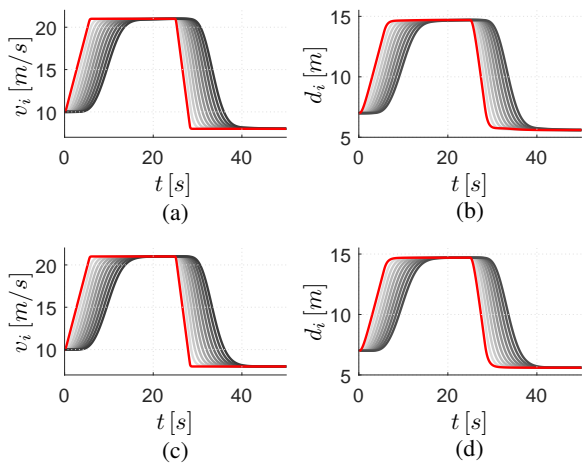}
\caption{Velocity and distance profiles for vehicle platoons controlled by $\bar{\mathcal{K}}$, in (a)-(b), and $\hat{\mathcal{K}}$, in (c)-(d), in case of ``attack-free" IVC. In red the speed of $\mathcal{V}_0$, in (a)-(c), and the relative distance between $\mathcal{V}_0$ and $\mathcal{V}_1$, in (b)-(d). From light grey to black are depicted speeds (relative distances) of vehicles with indexes from $1$ to $10$ ($2$ to $10$).} 
\label{fig:speedTest_noDrops}
\end{figure}

To conclude our numerical analysis, we gathered, in Table~\ref{tab:tunings}, the outcomes of Algorithm~\ref{al:mainAlgorithm} for different values of $h$, the constant time gap between vehicles. It emerges that the resiliency to DOS attacks increases with $h$.

\begin{table}[thpb]
\caption{Values of $\Delta$ and tuned parameters obtained for $\hat{\mathcal{K}}$  for different values of $h$ by using Algorithm~\ref{al:mainAlgorithm}.} 
\centering
\resizebox{\columnwidth}{!}{%
\begin{tabular}{|c||c|c|c|c|c|c|c|c|}
\hline
$h$ [s] & $0.4$ & $0.5$ & $0.6$ & $0.7$ & $0.8$ & $0.9$ & $1$ & $1.1$  \\ \hline
$\Delta$ & $1$  & $2$ & $4$ & $5$ & $6$ & $7$ & $8$ & $9$   \\ \hline
$k_p$ & $0.5$ & $0.5$ & $1.05$  & $0.82$ & $0.69$ & $0.59$ & $0.52$ & $0.46$  \\ \hline
$k_d$ & $1.73$ & $1.73$ & $3.23$ & $2.6$ & $2.25$ & $1.97$ & $1.78$ & $1.62$\\ \hline
\end{tabular}
}
\label{tab:tunings}
\end{table}

\section{Conclusion}
This paper proposed a hybrid controller for string stable homogeneous vehicle platoons. In particular, the proposed controller and tuning algorithm provides a tool to design a DOS-resilient CACC that also satisfies performance requirements. In addition, the tuning algorithm returns a metric to evaluate the resiliency to DOS attacks. Indeed, our approach allows estimating the maximum number of consecutive packet dropouts occurring during the DOS attacks that the proposed CACC can tolerate without losing string stability of the vehicle platooning.

The effectiveness of our approach has been shown in some numerical examples.
Our approach turns out having a higher resilience to DOS attacks compared to \cite{ploeg2011design}. Furthermore, since string stability is guaranteed with inter-vehicle time gaps smaller than one second, our approach is also more efficient compared to ACC or control approaches that rely only on onboard sensors \cite{ploeg2014graceful,wu2019cooperative,harfouch2017adaptive}.

Future research directions aim at extending the proposed approach to account for: heterogeneous platoons, measurement noise, and control input saturation.

\begin{figure}[bhp]
\vspace{5pt}
\centering
\psfrag{x1}[][][0.9]{$t \, [s]$}
\psfrag{a}[][][0.9]{(a)}
\psfrag{y1}[][][0.9]{$v_i \, [m/s]$}
\psfrag{x2}[][][0.9]{$t \, [s]$}
\psfrag{b}[][][0.9]{(b)}
\psfrag{y2}[][][0.9]{$d_i \, [m]$}
\psfrag{x3}[][][0.9]{$t \, [s]$}
\psfrag{c}[][][0.9]{(c)}
\psfrag{y3}[][][0.9]{$v_i \, [m/s]$}
\psfrag{x4}[][][0.9]{$t \, [s]$}
\psfrag{d}[][][0.9]{(d)}
\psfrag{y4}[][][0.9]{$d_i \, [m]$}
\includegraphics[scale=0.6]{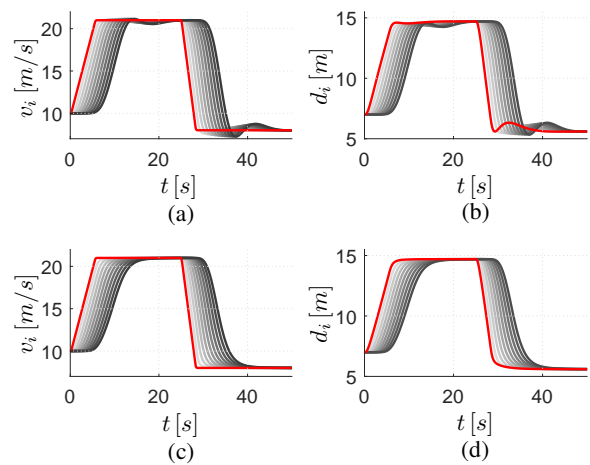}
\caption{Velocity and distance profiles for vehicle platoons controlled by $\bar{\mathcal{K}}$, in (a)-(b), and $\hat{\mathcal{K}}$, in (c)-(d), in case of IVC under DOS attacks. In red the speed of $\mathcal{V}_0$, in (a)-(c), and the relative distance between $\mathcal{V}_0$ and $\mathcal{V}_1$, in (b)-(d). From light grey to black are depicted speeds (relative distances) of vehicles with indexes from $1$ to $10$ ($2$ to $10$).} 
\label{fig:speedTest_Drops}
\end{figure}

% or


%\appendices
%% you can choose not to have a title for an appendix
%% if you want by leaving the argument blank
%\section{}

%\section*{Acknowledgment}
%This material is based upon work supported by the National Science Foundation (NSF) under grant No. CNS-1544910. Any opinions, findings and conclusions or recommendations expressed in this material are those of the authors and do not necessarily reflect the views of the National Science Foundation.

%\ifCLASSOPTIONcaptionsoff
%  \newpage
%\fi


\balance
\bibliographystyle{IEEEtran}
\bibliography{IEEEabrv,references}

%@IEEEtranBSTCTL{IEEEexample:BSTcontrol,
%CTLdash_repeated_names = "no"
%}

\balance
%\begin{IEEEbiography}{Roberto Merco}
\begin{IEEEbiography}[{\includegraphics[width=1in,height=1.25in,clip,keepaspectratio]{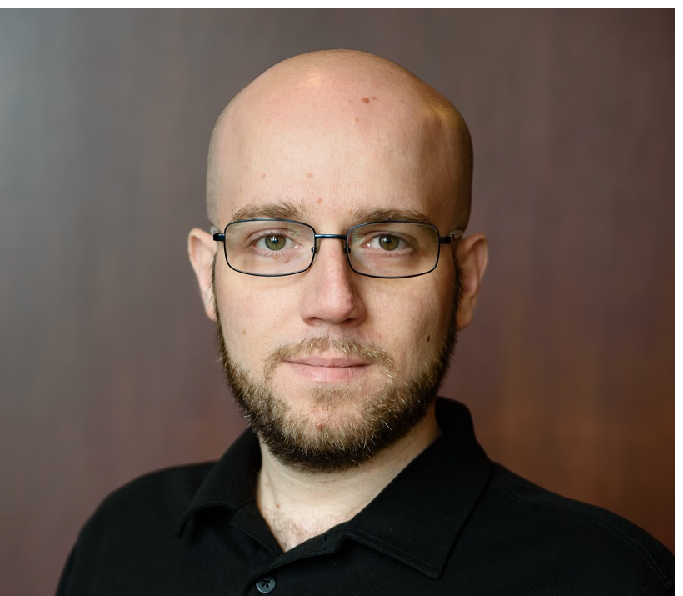}}]{Roberto Merco}
received the Ph.D. degree in Automotive Engineering from Clemson University in 2019, the B.Sc. degree cum laude in Control Engineering from University Tor Vergata, Rome, Italy in 2007, and the M.Sc. degree cum laude in Control Engineering from University Tor Vergata, Rome, Italy in 2010. From 2010 to 2016 he was working in a system integrator company in the field of industrial automation. His research interests include the area of controls of resilient networked control systems, hybrid dynamical systems, and intelligent transport systems.
\end{IEEEbiography}

%\begin{IEEEbiography}{Francesco Ferrante}
% ``"
\begin{IEEEbiography}[{\includegraphics[width=1in,height=1.25in,clip,keepaspectratio]{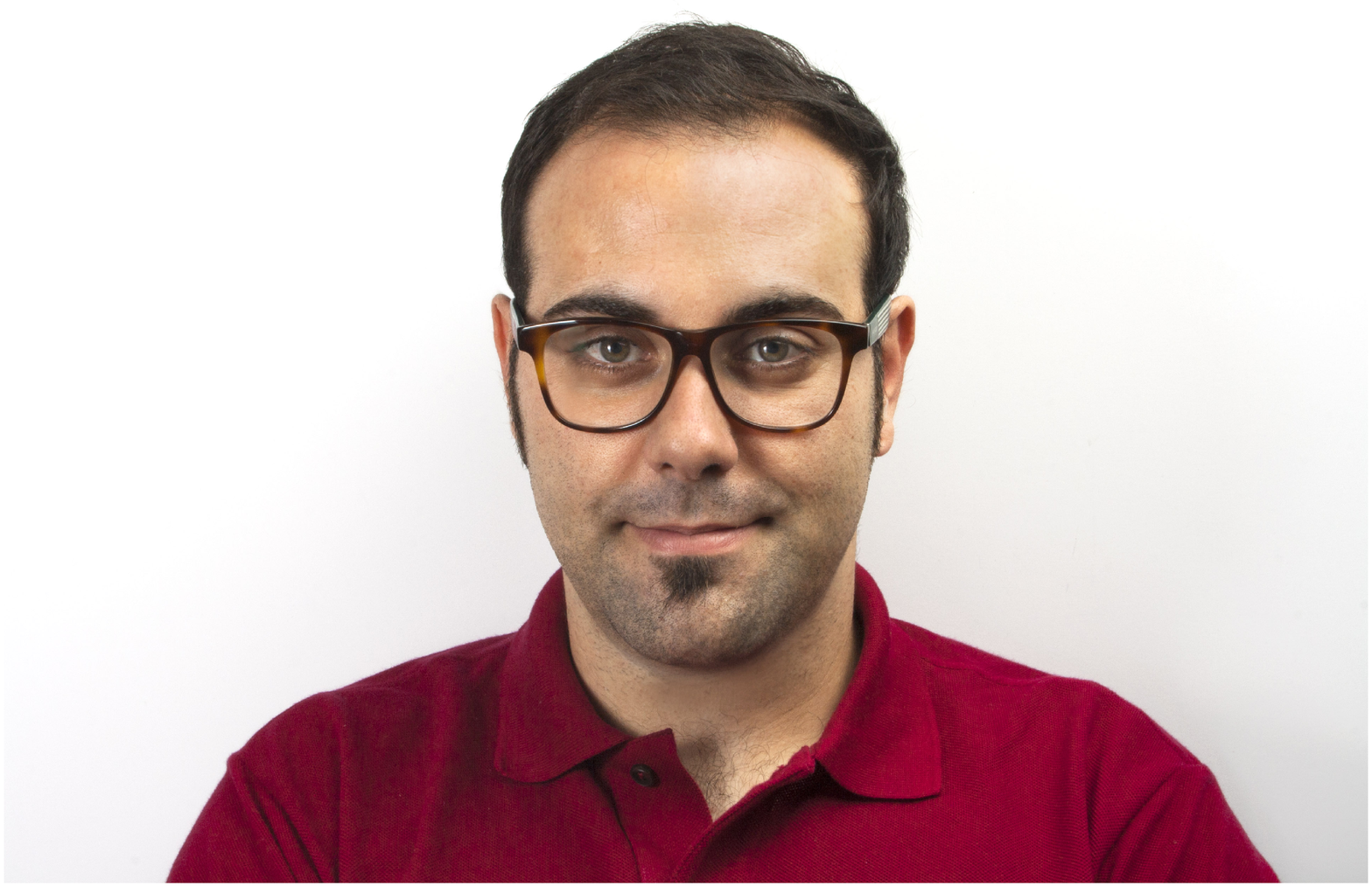}}]{Francesco Ferrante}
is an assistant professor at the Faculty of Sciences of the University of Grenoble Alpes, France. He also holds an adjunct assistant professor position at the Department of Automotive Engineering of Clemson University, USA.
He received in 2010 a ``Laurea degree" (BSc) in Control Engineering from University ``Sapienza" in Rome, Italy and in 2012 a ``Laurea Magistrale" degree (MSc) with honors in Control Engineering from University ``Tor Vergata" in Rome, Italy. During 2014, he held a visiting scholar position at the Department of Computer Engineering, University of California Santa Cruz. In 2015, he received a PhD degree in control theory from ``Institut supérieur de l'aéronautique et de l'espace" (SUPAERO) Toulouse, France. From November 2015 to August 2016, he was a postdoctoral fellow at the Department of Electrical and Computer Engineering, Clemson University. From August 2015 to September 2016, he held a position as postdoctoral scientist at the Hybrid Systems Laboratory (HSL) at the University of California at Santa Cruz. He currently serves as an associate editor in the conference editorial boards of the IEEE Control Systems Society and the European Control Association. 
\end{IEEEbiography}

%\newpage

%\begin{IEEEbiography}{Pierluigi Pisu}
\begin{IEEEbiography}[{\includegraphics[width=1in,height=1.25in,clip,keepaspectratio]{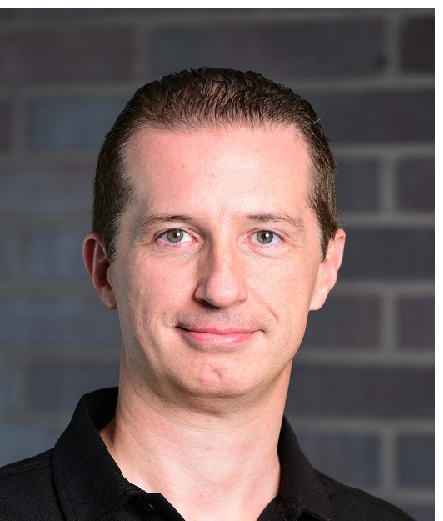}}]{Pierluigi Pisu}
received the Ph.D. degree in electrical engineering from The Ohio State University in 2002 and the Laurea in computer engineering from the University of Genoa, Italy. He is an Associate Professor of automotive engineering with Clemson University, with a joint appointment in the Holcombe Department of Electrical and Computer Engineering. He is the Leader of the Deep Orange 10 Program. He is the Director of the DOE GATE Hybrid Electric Powertrain Laboratory and the Creative Car Laboratory. His research interests include the area of functional safety, security, control and optimization of cyber-physical systems for next generation of high performance and resilient connected and automated systems with emphasis in both theoretical formulation and virtual/hardware-in-the-loop validation.
\end{IEEEbiography}

\end{document}